\documentclass[final,oribibl]{llncs}
\usepackage[applemac]{inputenc}
\usepackage{amsmath}
\usepackage{amssymb}
\usepackage{theorem}
\usepackage{graphicx}

\usepackage{lmodern}
\usepackage{fontenc}[T1]
\usepackage[english]{babel}
\usepackage{xspace}

\usepackage[ruled]{algorithm}
\usepackage[noend]{algpseudocode}

\usepackage[final]{hyperref}

\providecommand{\Ocomp}{\ensuremath{\mathcal{O}}}

\DeclareMathOperator{\weight}{w}

\DeclareMathOperator{\cl}{cl}

\usepackage{prettyref}
\newcommand{\prref}[1]{\prettyref{#1}}

\newrefformat{thm}{Theorem~\ref{#1}}
\newrefformat{lem}{Lemma~\ref{#1}}
\newrefformat{dfn}{Definition~\ref{#1}}
\newrefformat{cor}{Corollary~\ref{#1}}
\newrefformat{prop}{Proposition~\ref{#1}}
\newrefformat{sec}{Section~\ref{#1}}
\newrefformat{kap}{Chapter~\ref{#1}}
\newrefformat{ex}{Example~\ref{#1}}
\newrefformat{rem}{Remark~\ref{#1}}
\newrefformat{fig}{Figure~\ref{#1}}
\newrefformat{app}{Appendix~\ref{#1}}
\newrefformat{eq}{(\ref{#1})}
\newrefformat{tab}{Table~\ref{#1}}

\theoremstyle{plain}

\newcommand{\letters}{\ensuremath{B}}
\newcommand{\variables}{\ensuremath{\Omega}}
\newcommand{\morphism}{\ensuremath{h}}

\newcommand{\ie}{i.e.,\xspace}
\newcommand{\eg}{e.g.,\xspace}

\newcommand{\IFF}{if and only if\xspace}
\newcommand{\homo}{homomorphism\xspace}
\newcommand{\morph}{morphism\xspace}

\newcommand{\set}[2]{\left\{#1\mathrel{\left|\vphantom{#1}\vphantom{#2}\right.}#2\right\}}

\newcommand{\oneset}[1]{\left\{\mathinner{#1}\right\}}
\newcommand{\os}{\oneset}
\newcommand{\sm}{\setminus}

\newcommand{\sse}{\subseteq}

\newcommand{\vdmatrix}[4]{\left(\begin{smallmatrix}#1 & #2\\ #3 & #4\end{smallmatrix}\right)}

\newcommand{\abs}[1]{\left|\mathinner{#1}\right|}
\newcommand{\Abs}[1]{\left\Vert\mathinner{#1}\right\Vert}

\newcommand{\N}{\ensuremath{\mathbb{N}}}
\newcommand{\Z}{\ensuremath{\mathbb{Z}}}

\newcommand{\B}{\ensuremath{\mathbb{B}}}

\newcommand{\M}{\ensuremath{\mathbb{M}}}

\newcommand{\PSPACE}{\ensuremath{\mathsf{PSPACE}}}
\newcommand{\npoly}{\ensuremath{\mathsf{npoly}}}
\newcommand{\poly}{\ensuremath{\mathsf{poly}}}
\newcommand{\NPSPACE}{\ensuremath{\mathsf{NPSPACE}}}
\newcommand{\EXPSPACE}{\ensuremath{\mathsf{EXPSPACE}}}

\newcommand{\NP}{\ensuremath{\mathsf{NP}}}

\newcommand{\Bn}{\B^{m\times m}}
\newcommand{\Mn}{{\M}_{2m}}

\renewcommand{\phi}{\varphi}

\newcommand{\sig}{\sigma}

\newcommand{\sol}[1]{\ensuremath{\sig(#1)}}
\newcommand{\mysolution}{\ensuremath{\sig}}

\newcommand{\Sig}{\Sigma}
\newcommand{\Gam}{\GG}

\newcommand\GG{\Gamma}

\newcommand\OO{\Omega}

\newcommand{\cS}{\mathcal{S}}

\newcommand{\ov}[1]{\overline{#1}}
\newcommand{\inv}[1]{\ov{#1}}
\newcommand{\oi}[1]{{#1}^{-1}}
\newcommand{\invol}{\overline{\,^{\,^{\,}}}}

\newcommand{\algofont}[1]{\textnormal{\textsc \selectfont\sffamily  #1}}
\newcommand{\algncr}{\algofont{CompNCr}}
\newcommand{\algcr}{\algofont{CompCr}}

\newcommand{\algonephase}{\algofont{TransformEq}}

\begin{document}

\pagestyle{plain}

\title{Finding All Solutions of Equations in Free Groups and Monoids with Involution\titlerunning{Solutions of Equations}}

\author{%
  Volker Diekert\inst{1} \and
   Artur Je\.z\inst{2,3}\fnmsep\thanks{Supported by Humboldt Research Fellowship for Postdoctoral Researchers}
	\and
   Wojciech Plandowski\inst{4} 
   }
\authorrunning{V.~Diekert, A.~Je\.z, W.~Plandowski}

\institute{%
  Institut f\"ur Formale Methoden der Informatik,
  University of Stuttgart, Germany \and 
  Institute of Computer Science, University of Wroclaw, Poland \and
	Max Planck Institute f\"ur Informatik, Saarbr\"ucken, Germany \and
	Institute of Informatics, University of Warsaw, Poland}

\maketitle
\begin{abstract}
  The aim of this paper is to present a \PSPACE{} algorithm which yields a finite 
  graph of exponential size and which 
  describes the set of all solutions of equations in free groups as well as
	the set of all solutions of equations in free monoids with involution in the presence of 
  rational constraints. This became possible due to the recently invented \emph{recompression} technique of the second author.

He successfully applied the recompression technique for pure word equations without involution or rational constraints.
In particular, his method could not be used as a~black box for free groups (even without rational constraints).
Actually, the presence of an involution (inverse elements) and rational constraints complicates the situation
and some additional analysis is necessary.
Still, the recompression technique is general enough to accommodate both extensions.
In the end, it simplifies proofs that solving word equations is in \PSPACE{} (Plandowski 1999)
and the corresponding result for equations in free groups with rational constraints (Diekert, Hagenah and Guti{\'e}rrez 2001).
As a byproduct we obtain a~direct proof that it is decidable in \PSPACE{} whether or not the solution set is finite.
\footnote{A preliminary version of this paper was presented as an invited talk at CSR 2014 in Moscow, June 7--11, 2014.}
\end{abstract}

\section*{Introduction}\label{sec:intro}
A word equation is a simple object. It consists of a pair $(U,V)$ of words over constants and variables and a solution is a substitution of the variables by words in constants such that $U$ and $V$ become identical words. 
The study of word equations has a long tradition. Let \emph{WordEquation} be the 
problem of deciding whether a given word equation has a solution. 
It is fairly easy to see that WordEquation reduces to Hilbert's 10th Problem
(in Hilbert's famous list presented in 1900 for his address
at the International Congress of Mathematicians).
Hence in the mid 1960s the Russian school of mathematics outlined the roadmap to prove undecidability of  Hilbert~10 th Problem
via undecidability of WordEquation.
The program failed in the sense that Matiyasevich proved Hilbert's 10th Problem to be undecidable in 1970, but by a completely different method, which employed number theory.
The missing piece in the proof of the undecidability of Hilbert's 10th Problem
was based on methods due to Robinson, Davis,  and Putnam~\cite{mat93}.
On the other hand, in 1977 Makanin showed in a seminal paper~\cite{mak77} that WordEquation is decidable! 
The program went a different way, but its outcome were two major achievements in mathematics. Makanin's algorithm became famous since it  settled a long standing problem and also because his algorithm had an extremely complex termination proof. In fact, his paper showed that the 
existential theory of equations in free monoids is decidable. This is close to the borderline of decidability
as already the $\forall\exists^{3}$ positive theory of free monoids is undecidable~\cite{dur95}. 
Furthermore Makanin extended his results to free groups and showed that the existential and positive theories in free groups are decidable~\cite{mak82,mak84}.
Later Razborov was able in \cite{raz87} (partly shown also in \cite{raz93}) to describe the set of all solutions for systems of equations in free groups (see also \cite{KMII98} for a  description of Razborov's work). This line of decidability results culminated in the proof of Tarski's conjectures by Kharlampovich and
Myasnikov  in a series of papers ending in \cite{KMIV06}. In particular, they  showed that the  theory of free groups is decidable. In order to prove this fundamental result the description of  all solutions
of an equation in a free group is crucial.

Another branch of research was to extend Makanin's result to more general algebraic structures including free partially commutative monoids \cite{mat97lfcs,dmm99tcs}, free partially commutative monoids with involution, graph groups (also known as right-angled Artin groups) \cite{dm06}, graph products \cite{DiekertLohrey08}, and hyperbolic groups \cite{rs95,DahmaniGui10}. 
In all these cases the existential theory of equations is decidable. 
Proofs used the notion of \emph{equation with rational constraints}, which was first developed in the habilitation of Schulz~\cite{sch91}.
The concept of equation with rational constraints is used also throughout the present paper.

In parallel to these developments there were drastic improvements in the 
complexity of deciding Wordequation. It is fairly easy to see that 
the problem is \NP-hard. Thus, \NP{} is a lower bound. 
First estimations for the time complexity on Makanin's algorithm for free monoids led to a tower of several exponentials, but it was lowered over time to \EXPSPACE{} in \cite{gut98focs}. On the the other hand it was shown in \cite{kp96} that 
Makanin's scheme for solving equations in free groups is not primitive recursive.
(Already in the mid 1990 this statement was somehow puzzling and counter-intuitive, as it suggested a strange crossing of complexities:
The  existential theory in free monoids seemed to be easier than the one in free groups,
whereas it was already known at that time that the positive theory in free monoids is undecidable, but decidable in free groups.)
The next important step was done by Plandowski and Rytter,
whose approach~\cite{pr98icalp} was the first essentially different than Makanin's original solution.
The main idea was to apply compression to WordEquation and the result was that the length-minimal solution
of a word equation compresses well, in the sense that {L}empel-{Z}iv encoding, which is a popular practical standard of compression,
of such a solution is exponentially smaller than the solution itself (if the solution is at least exponential in the length of the equation).
This yielded an $\npoly(n,\log N)$ algorithm for WordEquation,
note that at that time the only available bound on $N$ was the triply exponential bound by Makanin.
Still, this result prompted Plandowski and Rytter to formulate a (still open) conjecture that WordEquation is \NP-complete.

Soon after a doubly exponential bound on $N$ was shown by Plandowski~\cite{Pl2},
this bound in particular used the idea of representing the solutions in a compressed form (in fact, the equation as well is kept in a compressed form)
as well as employing a novel type of factorisations.
Exploiting better the interplay between factorisations and compression
Plandowski showed that WordEquation is in \PSPACE{}, \ie
it can be solved in polynomial space and exponential time~\cite{pla04jacm}.
His method was quite different from Makanin's approach and more symmetric.
Furthermore, it could be also used  to generate all solutions of a given word equation~\cite{Pl5},
however, this required non-trivial extensions of the original method.

Using Plandowski's method Guti{\'e}rrez showed that satisfiability of equations in free groups is in \PSPACE{} \cite{gut2000stoc},
which led Diekert, Hagenah and Guti{\'e}rrez to the result that the existential theory of equations with rational 
constraints in free groups is \PSPACE-complete \cite{dgh05IC}.
Without constraints \PSPACE{} is still the best upper bound, although the existential theories for equations in free monoids (with involution) and
free groups are believed to be \NP{} complete.
Since this proof generalized Plandowski's satisfiability result~\cite{pla04jacm},
it is tempting to also extend the generator of all solutions~\cite{Pl5}.
Indeed, Plandowski claimed that his method applies also to free groups with rational constraints, but he found a gap in his generalization~\cite{Pl6}.

However in 2013 another substantial progress in solving word equations was done due to
a powerful recompression technique by Je\.z~\cite{wordequations}. His new proof that WordEquation is in \PSPACE{} 
simplified the existing proofs drastically.
In particular, this approach could be used to describe the set of all solutions rather easily,
so the previous construction of Plandowski~\cite{Pl5} was simplified as well.
 
What was missing however was the extension to include free monoids with involution and therefore free groups and another missing block was the 
the presence of rational constraints. 
Both extensions are the subject of the present paper. 

\subsubsection*{Outline}
We first follow the approach of~\cite{dgh05IC} how to (bijectively) transform 
the set of all solutions of an equation with rational constraints over a free group in polynomial time
into a set of all solutions of an equation with regular constraints over a free monoid with involution,
see Section~\ref{subsec: free groupmonoid}.
Starting at that point in Section~\ref{sec:groas} we formulate the main technical claim of the paper:
existence of a procedure that transforms equations over the free monoid and (roughly speaking) keeps the set of solutions
as well as does not increase the size of the word equation;
in particular in this section we make all the intuitive statements precise.
Moreover, we show how this procedure can be used to create
\PSPACE{}-transducer which produces a finite graph
(of exponential size) describing all solutions and which is nonempty \IFF the equation has at least one solution. Moreover, the graph also 
encodes whether or not there are finitely many solutions, only. 
The technique of recompression simplifies thereby~\cite{dgh05IC} and it yields the important new feature that we can describe all solutions.

\section{Preliminaries}
As already mentioned, the general plan is to reduce the problem of word equation with regular constraints over free group
to the problem of word equation with regular constraints over a free monoid with an involution
and give an algorithm for the latter problem.
In this section we first introduce all notions regarding the word equation over the free monoid, see Section~\ref{subsec: free monoid},
and only afterwards the similar notions for a free group together with the reduction of the latter scenario to the former one,
see Section~\ref{subsec: free groupmonoid}.

\subsection{Word equations over a free monoid with involution}
\label{subsec: free monoid}
Let $A$ and $\OO$ be two finite disjoint sets, called
{\em the alphabet of constants} and {\em the alphabet of variables} (or \emph{unknowns}), respectively.
For the purpose of this paper $A$ and $\OO$ are endowed with an \emph{involution}, which is
is a mapping $\invol$ such that 
$\overline{\overline{x}} = x$ for all elements.
In particular, an involution is a~bijection.
If involution is defined for a monoid, then  we additionally require that $\overline{xy}=\overline{y}\,\overline{x}$ for all its elements $x, y$.
This applies in particular to a free monoid $A^*$
over a set with involution: For a word $w = a_1 \cdots a_m$  we thus have $\ov{w} = \ov{a_m} \cdots \ov{a_1}$. If $\ov a = a$ for all $a \in A$
then $\ov{w}$ simply means to read the word from right-to-left. 
It is sometimes useful to consider \emph{involution closed sets}, \ie such that $\inv S = S$.

A {\em word equation} is a pair $(U,V)$
of words over $A\cup\OO$, often denoted by $U=V$.
A {\em solution} \mysolution{}
of a word equation $U=V$ is a substitution $\sig$ of unknowns in $\OO$ by
words over constants, such that the replacement of unknowns by the substituted words in $U$ and in $V$ give the same word.
Moreover, as we work with involutions we additionally demand that the solution satisfies 
$\sig(\ov X) = \ov{\sig(X)}$ for all $X \in \OO$.
If an equation 
does not have simultaneous occurrences of $X$ and $\ov X$ where 
$X \neq \ov X$ then this additional requirement is vacuous. 
A solution is \emph{non-empty}, if $\sol X \neq \epsilon$ for every variable $X$ such that $X$ or $\inv X$ occurrs in the equation.
During the proof we will consider only non-empty solutions.
This is non-restrictive, as we can always non-deterministically guess the variables that are assigned $\epsilon$ by a solution
and remove such variables form the equation.
On the other hand, it is useful to assume that a solution assigns $\epsilon$ to each variable $X$ such that $X$, nor $\inv X$ occur in the equation:
during the algorithm we remove the variables that are assigned $\epsilon$ in the solution.
Nevertheless, we need to know the substitution for such $X$, as we create the set of all solutions by backtracking.

\begin{example}
Let $\OO=\os{X,Y, \ov X, \ov Y}$ and $A=\os{a,b}$ with  $b = \ov a$. 
Then $XabY=YbaX$ behaves as a  word equation without involution 
One of its solutions is the substitution
$\sig(X)=bab$, $\sig(Y)=babab$. Under this substitution we have $\sig(X)ab \sig(Y)=bababbabab=\sig(Y)ba \sig(X)$. It can be proved that 
the solution set of the equation $XabY=YbaX$ is closely related to Sturmian words \cite{IlPl}.
\end{example}

The notion of word equation immediately generalizes to a system of word 
equations $(U_1, V_1), \ldots, (U_s, V_s)$.
In this case a solution $\sig$ must satisfy all $(U_i, V_i)$ simultaneously.
However, such a system can be reduced to a single equation
$(U_1a \cdots U_s a U_1b \cdots U_s b, V_1a \cdots V_s a V_1b \cdots V_s b)$
where $a$, $b$ are fresh constants with $a \neq b$.
Furthermore, this reduction remains valid when additionally regular constraints are introduced,
such constraints are properly defined below.

Lastly, we always assume that the involution on $\OO$ is without fixed points:
otherwise for a variable $X$ such that $\inv X = X$ we can introduce a fresh variable $X'$, set $\inv X = X'$ and add an equation $X = X'$,
which ensures that $\sol X = \inv{\sol X}$.
In this way we can avoid some case distinctions.

\subsubsection*{Constraints}
Let $\mathcal C$ be a class of formal languages, then a system of 
{\em word equations with constraints in $\mathcal C$} is given by a finite list $(U_i,V_i)_{i}$
of word equations and a finite list of constraints of type 
$X \in L$ (resp.{} $X \notin L$) where $X \in \OO$ and $L \sse A^*$ with $L \in \mathcal C$.
For a solution we now additionally demand that $\sig(X) \in L$ 
(resp.{} $\sig(X) \notin L$) for all constraints.

Here, we focus on rational and recognizable (or regular) constraints and 
we assume that the reader is familiar with basic facts in formal language theory. 
The classes of rational and recognizable subsets are
defined for every monoid $M$~\cite{eil74}, and they are incomparable, in general.  \emph{Rational} sets (or
languages) are defined inductively as follows. 
\begin{itemize}
	\item All finite subsets of $M$ are rational.
	\item If $L_1,L_2 \subseteq M$ are rational, then the union $L_1 \cup L_2$, the concatenation $L_1 \cdot L_2$, and the generated submonoid 
$L_1^*$ are rational.
\end{itemize}
A subset
$L \subseteq M$ is called \emph{recognizable}, if  there is a
homomorphism $\rho$ to some finite monoid $E$ such that $L =
\rho^{-1}\rho(L)$.  We also say that $\rho$ (or $E$) \emph{recognizes} $L$ in this case. 
Kleene's Theorem states that in finitely generated
free monoids both classes coincide, and we follow the usual convention
to call a rational
subset of a free monoid {\em
  regular}. If $M$ is  generated by some finite set $\Gam \sse M$  (as it always the case in this paper) then every rational set is the image 
  of a regular set $L$ under the canonical \homo from $\Gam^*$ onto $M$;
  and every recognizable set of $M$  is rational. 
  (These statements are trivial consequences of Kleene's Theorem.)   
 Therefore,  
  throughout we assume that a rational (or regular) language is specified by a nondeterministic finite automaton, \emph{NFA} for short.

Consider a list of  $k$ regular languages $L_i \sse \Sig^*$ each of them being  specified 
by some NFA with  $m_i$ states. The disjoint union of these automata 
yields a single NFA with $m = m_1 + \cdots + m_k$ states which accepts all $L_i$
by choosing appropriate initial and final sets for each $L_i$;
we may assume that the NFA has state set $\os{1,\ldots m}$. Then each constant $a\in A$ defines a 
Boolean $m\times m$ matrix $\tau(a)$ where the entry $(p,q)$ is $1$ if
$(p,a,q)$ is a transition and $0$ otherwise. This yields a 
\homo $\tau : A^* \to \Bn$ such that $\tau$ recognizes $L_i$ for all 
$1 \leq i \leq k$.

Moreover, for each $i$ there is a row vector 
$I_i \in \B^{1\times n}$ and a column vector $F_i \in \B^{n\times 1}$
such that we have $w \in L_i$ \IFF $I_i \cdot \tau(w) \cdot F_i = 1$. 

For a matrix $P$ we let $P^{T}$ be its transposition.
There is no reason that  $\tau(\ov a)=\tau(a)^{T}$,
hence $\tau$ is not necessarily a homomorphism which respects the involution.
So, as done in~\cite{dgh05IC}, we  let $\Mn \subseteq \B^{2m \times 2m}$
denote the following monoid with
involution:
$$\Mn= \set{\vdmatrix P00Q}{ P,Q \in \Bn} \text { with }  \overline{\vdmatrix  P00Q} =
\vdmatrix {Q^T}00{P^T}.
$$
Define  $\rho(a) = \vdmatrix{\tau(a)}{0}{0}{\tau(\ov a)^{T}}$. Then the \homo 
$\rho : A^* \to \Mn$ respects the involution. Moreover 
$\rho$ recognizes all $L_i$ and $\ov {L_i} =  \set{\ov w}{w \in L_i}$.

Consider regular constraints $X \in L$ and $X \notin L'$.
As $\rho$ recognises both $L$ and $L'$, the conditions $\sol X \in L$ and $\sol X \notin L'$
are equivalent to $\rho(\sol X) \in \rho(L)$ and $\rho(\sol X) \notin \rho(L')$.
As the image of $\rho$ is a subset of $\Mn$, there are only finitely many elements in it.
Thus all regular constraints on $X$ boil down to restrictions of possible values of $\rho(\sol X)$.
To be more precise, if all positive constraints on $X$ are $(L_i)_{i \in I}$ and all negative are
$(L_i')_{i \in I'}$, all those constraints are equivalent to
$$
\rho(\sol X) \in \bigcap_{i \in I} \rho(L_i) \cap \bigcap_{i \in I'} (\Mn \setminus \rho(L_i')) \enspace .
$$
Thus, as a preprocessing step our algorithm guesses the $\rho(\sol X)$, which we shall shortly denote as $\rho(X)$,
moreover this guess needs to satisfy
\begin{itemize}
	\item $\rho(\overline{X})=\overline{\rho(X)}$
	\item $\rho(X) \in \rho(L)$ for each positive constraint $L$ on $X$;
	\item $\rho(X) \notin \rho(L')$ for each negative constraint $L'$ on $X$.	
\end{itemize}
In the following we are interested only in solutions for which $\rho(\sol X) = \rho(X)$.
Note that, as $\rho$ is a function, each solution of the original system corresponds to a solution for an exactly one such a guess,
thus we can focus on generating the solutions for this restricted problem.

We now give a precise definition of the main problem
we are considering in the rest of the paper:
\begin{definition}
  An {\em equation $E$ with constraints} is a tuple
$E=(A,\Omega, \rho ; U=V)$
containing the following items:
\begin{itemize}
\item An alphabet of constants with involution $A$.
\item An alphabet of variables with involution without fixed points $\OO$.
\item A mapping $\rho: A\cup \Omega \to \Mn$ such that 
      $\ov{\sigma(x)} = \sig(\ov x)$ for all
      $x \in A \cup \Omega$.
\item The word equation $U=V$ where $U,V \in (A \cup \Omega)^*$.
\end{itemize}
A {\em solution} of  $E$
  is a 
  \homo $\sigma: (A \cup \Omega)^* \to A^*$
  leaving the constants from $A$ invariant such that the following
  conditions are satisfied:
\[
\begin{array}{rclll}
 \sigma(U) & = & \sigma(V)\, , && \\
  \ov{\sigma(X)} & = & \sig(\ov X) & \textrm{for all} & X\in \Omega,\\
 \rho(\sigma(X)) & = & \rho(X) & \textrm{for all} & X\in \Omega.\\
 \end{array}
\]
  The {\em input size} of $E$ is given by 
$\Abs E  = \abs A + \abs \OO + \abs {UV} + m$.
\end{definition}

In the following, when this does not cause a confusion, we denote both the size of the instance and the length of the equation by $n$.
Note that we can always increase the size of the equation by repeating it several times.

The measure of size of the equation is accurate enough with respect to polynomial 
time and/or space. For example note that if an NFA has $m$ states then 
the number of transitions is bounded by $m \abs A$. Note also that
$\abs A$ can be much larger than the sum over the lengths of the equations and inequalities plus
the sum of the number of states of the NFAs in the lists for the constraints.

As already noted, by a convention, when a variable $X$ and its involution $\inv X$ are not present in the equation,
each solution assigns $\epsilon$ to both $X$ and $\inv X$.
In particular, this assignment should satisfy the constraint, \ie $\rho(X) = \rho(\epsilon)$
for each variable not present in the solution.
Note that the input equation can have variables that are not present in the equation and have constraints other than $\rho(X) = \rho(\epsilon)$,
however, such a situation can be removed by a simple preprocessing.

\subsubsection{Equations during the algorithm.}
During the procedure we will create various other equations and introduce new constants.
Still, the original alphabet $A$ never changes and new constants shall represent words in $A^*$.
As a consequence, we will work with equations over $\letters \cup \variables$,
where $B$ is the smallest alphabet containing $A$ and all constants in $UV\ov{UV}$.
We shall call such $B$ the \emph{alphabet of} $(U, V)$.
Note that $\abs B \leq \abs A + 2 \abs {UV}$ and we therefore we can ignore  $\abs B$ for the complexity.

Ideally, a solution of $(U, V)$ assigns to variables words over the alphabet of $(U, V)$, call it $\letters$.
However, as our algorithm transforms the equations and solutions,
it is sometimes more convenient to allow also solutions that assign words from some $\letters' \supset \letters$.
A solution is \emph{simple} it if uses only constants from $\letters$,
by default we consider simple solutions.
Whenever we consider a non-simple solution, we explicitly give the alphabet over which this is a solution.

To track the meaning of constants outside $A$, we additionally require that a solution (over an alphabet $\letters'$)
supplies some homomorphism $\morphism: \letters' \mapsto A^*$, which is constant on $A$
and compatible with $\rho$, in the sense that $\rho(b) = \rho(\morphism(b))$ for all $b \in B$.
(Due to its nature, we also assume that $\morphism(b)$ contains at least two constants for $b \in \letters' \setminus A$.)
Thus, in the following, a solution is a pair $(\mysolution, \morphism)$.
In particular, given an equation $(U, V)$ the $\morphism (\sol U)$ corresponds to a solution of the original equation.

A \emph{weight} of a solution $(\mysolution,\morphism)$ of an an equation $(U,V)$ is
\begin{equation}
\label{eq: weight}
\weight(\mysolution,\morphism)= |U| + |V| + \sum_{X \in \variables} |UV|_X \abs{\morphism(\sol X)} \enspace ,
\end{equation}
where $|UV|_X$ denotes the number of occurrences of $X$ in $U$ and $V$ together.
The main property of such defined weight is that it decreases during the run of the algorithm,
using this property we shall guarantee a termination of the algorithm:
each next equation in the sequence will have a smaller weight, which ensures that we do not cycle.

Given a non-simple solution $(\mysolution,\morphism)$ we can replace all constants $c \notin \letters$
(where $\letters$ is the alphabet of the equation) in all \sol X by $\morphism(c)$
(note, that as $\rho(c) = \rho(\morphism(c))$, the $\rho(X)$ is preserved in this way).
This process is called a \emph{simplification} of a solution and the obtained substitution $\mysolution'$ is a simplification of \mysolution.
It is easy to show that $\mysolution'$ is a solution and that $\morphism(\mysolution'(U)) = \morphism(\mysolution(U))$,
so in some sense both \mysolution{} and $\mysolution'$ represent the same solution of the original equation.
Lastly, \mysolution{} and $\mysolution'$ have the same weight, see Lemma~\ref{lem: simplify the solution}.
Thus, in some sense we can always simplify the solution.

As a final note observe that $\morphism$ is a technical tool used in the analysis,
it is not stored, nor transformed by the algorithm,
nor it is used in the graph representation of all solutions.

\begin{lemma}
\label{lem: simplify the solution}
Suppose that $(\mysolution,\morphism)$ is a solution of the equation $(U, V)$.
Then the simplification $(\mysolution',\morphism)$ of $(\mysolution,\morphism)$ is also a solution of $(U, V)$,
$\morphism(\mysolution'(U)) = \morphism(\mysolution(U))$ and $\weight(\mysolution,\morphism) = \weight(\mysolution,\morphism)$.
\end{lemma}
\begin{proof}
Let $\letters$ be the alphabet of the equation and $\letters'$ the alphabet of the solution $\mysolution$.
Consider any constant $b \in \letters' \setminus \letters$.
As it does not occur in the equation, all its occurrences in \sol U and \sol V come from the varibles, \ie from some \sol X.
Then replacing all occurrences of $b$ in each \sol X by the same string $w$ preserves the equality of $\sol U = \sol V$,
thus $\mysolution'$ is also a solution.
Since we replace some constants $b$ with $\morphism(b)$ (and $\morphism \circ \morphism = \morphism$),
clearly $\morphism(\sol X) = \morphism(\mysolution'(X))$ for each variable.
Furthermore, as $\rho(c) = \rho(\morphism(c))$ we have that $\morphism(\sol X) = \morphism(\mysolution'(X))$.
Thus, $\morphism(\mysolution'(U)) = \morphism(\mysolution(U))$ and $\weight(\mysolution,\morphism) = \weight(\mysolution,\morphism)$,
as claimed.
\qed
\end{proof}

\subsection{Word equations with rational constraints over free groups.}
\label{subsec: free groupmonoid}
By $F(\Gam)$ we denote the free group over a finite set $\Gam$. We let 
$A= \Gam \cup \oi \Gam$. Set also $\ov x = \oi x$ for all $x \in F(\Gam)$.
Thus, in (free) groups we identify $\oi x$ and $\ov x$.
By a classical result of Benois \cite{ben69} 
rational subsets of $F(\Gam)$ form an effective Boolean algebra. 
That is: if $L$ is rational and specified by some NFA then  $F(\Gam)\sm L$ is rational; and we can effectively find the corresponding NFA. There might be an exponential blow-up in the NFA size, though. This is the main reason to allow 
  negative constraints $X \notin L$, so we can avoid explicit complementation.

\begin{proposition}[\cite{dgh05IC}]\label{prop:dghftom05}
Let $F(\Gam)$ be a free group and $A= \Gam \cup \oi \Gam$ be the corresponding set with involution as above. 
There is polynomial time transformation which takes as input a system 
$\cS$ 
of equations (and inequalities) with rational constraints over $F(\Gam)$ and outputs 
a word equation with regular constraints $\cS'$ over $A$ 
which is solvable \IFF $\cS'$ is solvable in $F(\Gam)$. 

More precisely, 
let $\phi: A^* \to F(\Gam)$ be the canonical \morph of the free monoid with involution $A^*$ onto the free $F(\Gam)$.
Then the set of all solutions for $\cS'$ is bijectively mapped via $\sig' \mapsto \phi \circ \sig'$ onto the set of all solutions of $\cS$. 
\end{proposition}

\prref{prop:dghftom05} in particular shows that the description of all solutions of a system of equations and inequalities (with rational constraints)
over a free group can be efficiently reduced to solving
the corresponding task for word equations with regular constraints in a free monoid with involution. For convenience of the reader let us remark that the proof of \prref{prop:dghftom05} is fairly straightforward. It is based on the fact that 
 $XYZ=1$ in $F(\Gam)$ is equivalent with the existence of words
 $P,Q,R \in A^*$ such that $X=P\ov Q$,  $Y=Q\ov R$,
 and $Z=R\ov P$. Indeed, if $XYZ=1$ in $F(\Gam)$ then we can represent
 $X$, $Y$, and $Z$ by reduced words and the existence of $P,Q,R$ follows because $F(\Gam)$ is a free group. The other direction is trivial and holds for non reduced words as well.

\subsubsection*{Input size.}
The input size for the reduction is given by the sum over the lengths of the equations and inequalities plus the size of $\Gam$ plus 
the sum of the number of states of the NFAs in the lists for the constraints.
As in the case of word equations over free monoid, the measure is accurate enough with respect to polynomial time and or space.
Note that $\abs \Gam$ can be much larger than the sum over the lengths of the equations and inequalities
plus the sum of the number of states of the NFAs in the lists for the constraints.
Recall that we encode $X\neq 1$ by a rational constraint, which introduces an NFA with $2\abs \Gam + 1$ states.
Since $\abs \Gam$ is part of the input, this does not cause any problem.
The output size remains at most quadratic in the input size.

\subsection{Existential theory for free groups}
We can easily extend the algorithm for word equations over free groups with rational constraints to existential theory
of free groups with rational constraints.
As a first step note that we can eliminate the disjunction by non-deterministic guesses.
Secondly, as the singleton $\os 1 \sse F(\Gam)$ is, by definition, rational,
the set $F(\Gam)\sm \os 1$ is rational, too. 
Therefore an inequality $U\neq V$ can be handled by a new fresh variable
$X$ and writing $U=X V \; \& \; X \in F(\Gam)\sm \os 1$ instead of $U\neq V$.

\subsection{Linear Diophantine systems}\label{sec:lDs}
We shall consider linear Diophantine systems with solutions over natural numbers.
Formally, such a system is given by an $m \times n$ matrix
$A$ with coefficients in $\Z$ and an $m\times 1$ vector $b \in \Z^m$. 
We write $Ax =b$ and its set of solutions is given by 
the set $\set{x \in \N^n}{Ax =b}$. We say that $Ax =b$ is satisfiable over $\N$ if the set of solutions is non-empty. 
Note that while we could also allow inequalities,
a system of inequalities $Ax \geq b$ can be reduced to equalities by introducing fresh variables $y$
and rewriting the system as $Ax - y = b$.
Looking for solutions in $\N^n$ makes the problem $\NP$-hard. 
Actually, we use the following well-known proposition.

\begin{proposition}\label{prop:lDs}
The following two problems are $\NP$-complete.\\
{\bf Input.} $Ax = b$ where $A \in \Z^{n\times n}$ and $b \in \Z^{n\times 1}$ and coefficients are written in binary.\\
{\bf Question 1.} Is the set $\set{x \in \N^n}{Ax =b}$ non-empty?\\
{\bf Question 2.} Is the set $\abs{\set{x \in \N^n}{Ax =b}}$ infinite?
\end{proposition}

\begin{proof}
 The $\NP$-completeness of the first problem is standard, see \eg \cite {HU}. It can be reduced to the second problem by adding an equation
 $y-z=0$ where $y, z$ are fresh variables. A possible reduction of the second problem to satisfiability is as follows. 
Given an equation $Ax = b$ we create a system $Ax = b \; \& \; Ax' = b$,
where $x = (x_1, \ldots, x_n)$ and $x' = (x_1', \ldots , x_n')$ use disjoint sets of variables.
Then we add equations $x' = x+ y$ where $y$ uses fresh variables. This guarantees that $x' \geq x$.
Finally, we add an equation $x_1' + \dots + x_n' = x_1 + \dots + x_n + z + 1$. 
This guarantees that $x'> x$, in the sense that at least one  of the inequalities $x_i' \geq x_i$ is strict.
If the new system is satisfiable then $Ax = b$ has infinitely many solutions $x + k(x'-x)$ with $k \in \N$.
Conversely, if $Ax = b$ has infinitely many solutions then there must exist solutions $x$ and $x'$ with $x < x'$
due to Dickson's Lemma~\cite{dickson1913}, and they satisfy the created system.
\end{proof}

\section{Graph representation of all solutions}\label{sec:groas}
In this section we give an overview of the graph representation of all solutions and the way such a representation is generated
as well as a detailed description of the graph representation of all solution of word equation with constraints.
This description is devised so that it is a citable reference,
in particular, it is supposed to be usable without reading the actual construction and the proof of its correctness.
It will include all the necessary definitions, though.
The actual construction and the proof of correctness is given in Section~\ref{sec: compression steps}.

\subsection{Transforming the equation}
By an \emph{operator} we denote a function that transforms substitutions (for variables).
All our operators have simple description:
$\mysolution'(X)$ is usually obtained from \sol X by morphisms, appending/prepending constants, etc.
In particular, they have a~polynomial description.
We usually denote them by $\phi$ and their applications by $\phi[\mysolution]$.

Recall that the instance size is $n$, so in particular the input equation is of length at most $n$ and has at most $n$ variables.
\begin{definition}
A word equation $(U, V)$ with constraints is \emph{strictly proper} if
\begin{itemize}
	\item in total $U$ and $V$ have at most $cn^2$ of constants;
	\item in total $U$ and $V$ have at most $n$ occurrences of variables;
	\item there is a homomorphism $\morphism : \letters \mapsto A^+$ that is compatible with $\rho$,
	where $\letters$ is the alphabet of $(U, V)$.
\end{itemize}	
An equation is \emph{proper} if instead of the first condition it satisfies a weaker one
\begin{itemize}
	\item in total $U$ and $V$ have at most $2cn^2$ constants.
\end{itemize}
\end{definition}

A possible constant is $c = 27$ as we will see later.
The idea is that strictly proper equations satisfy the desired upper-bound
and proper equations are some intermediate equations needed during the computation,
so they can be a bit larger.

Concerning the existence of \morphism, note that we do not want to consider equations
containing letters that cannot represent strings in the input alphabet.
For the input equation we may assume $A=B$ and therefore we can take \morphism{} as the identity.
The input equation is strictly proper.

The main technical result of the paper states that:
\begin{lemma}
\label{lem: main}
Suppose that $(U_0,V_0)$ is a strictly proper equation with $|U_0|, |V_0| > 0$ and let it have a solution
$(\mysolution_0,\morphism_0)$.
Then there exists a sequence of proper equations $(U_0, V_0)$, $(U_1, V_1)$, \ldots, $(U_k, V_k)$, over alphabets
$\letters_0$, $\letters_1$, \ldots, $\letters_k$ and solutions $(\mysolution_0, \morphism_0)$, $(\mysolution_1, \morphism_1)$,
\ldots, $(\mysolution_k, \morphism_k)$ of those equations and
families of operators $\Phi_1$, $\Phi_2$, \ldots, $\Phi_k$ and their simple solutions
 such that
\begin{itemize}
	\item $k > 0$ and $(U_k,V_k)$ is strictly proper.
	\item There is $\phi_{i+1} \in \Phi_{i+1}$ 	and a solution $(\mysolution_{i+1}',\morphism_{i+1})$ of $(U_{i+1},V_{i+1})$ over $B_i \cup B_{i+1}$
	such that
	\begin{itemize}
		\item $\mysolution_i = \phi_{i+1}[\mysolution_{i+1}']$
		\item $\mysolution_{i+1}$ is a simplification of $\mysolution_{i+1}'$
		\item $\morphism_i(\mysolution_i(U_i)) = \morphism_{i+1}(\mysolution_{i+1}(U_{i+1})) = \morphism_{i+1}(\mysolution_{i+1}'(U_{i+1}))$.
	\end{itemize}
	Furthermore, $\weight(\mysolution_i,\morphism_i) > \weight(\mysolution_{i+1},\morphism_{i+1}) = \weight(\mysolution_{i+1}',\morphism_{i+1})$.
	\item If $(\mysolution_{i+1}', \morphism_{i+1}')$ is a solution of $(U_{i+1},V_{i+1})$ (over an arbitrary alphabet)
	and $\phi_{i+1} \in \Phi_{i+1}$ then $(\mysolution_i', \morphism_i')$ is a solution of $(U_i,V_i)$,
	where $\mysolution_i' = \phi_{i+1}[\mysolution_{i+1}']$ and $\morphism_i'$ is some homomorphism compatible with $\rho$.
	\item Each family $\Phi_i$ as well as operator $\phi_i \in \Phi_i$ have polynomial-size description.
\end{itemize}
Given $(U_0,V_0)$, all such sequences (for all possible solutions) can be produced in \PSPACE.
\end{lemma}

\paragraph{Discussion}
The exact definition of allowed families of operators $\Phi$ is deferred to Section~\ref{subsec: graph representation},
for the time being let us only note that $\Phi$ has polynomial description
(which can be read from $(U_i,V_i)$ and $(U_{i+1},V_{i+1})$),
may be infinite and its elements can be efficiently listed,
(in particular, it can be tested, whether $\Phi$ is empty or not).

Concerning the difference between $\mysolution_{i+1}$ and $\mysolution_{i+1}'$:
while we know that $\mysolution_{i} = \phi_{i+1}[\mysolution_{i+1}']$ we cannot guarantee that $(\mysolution_{i+1},\morphism_{i+1})$ is simple,
so the claim of the Lemma~\ref{lem: main} does not apply to it $\mysolution_{i+1}'$.
However, when we take a simplification $(\mysolution_{i+1},\morphism_{i+1})$ of $\mysolution_{i+1}'$, the claim applies.
Moreover, $\mysolution_{i+1}$ and $\mysolution_{i+1}'$ represent the same solution
$\morphism_{i+1}(\mysolution_{i+1}(U_{i+1})) = \morphism_{i+1}(\mysolution_{i+1}'(U_{i+1}))$
of the original equation, so nothing is lost in the substitution.
Alternatively, we could impose the condition that the solution $(\mysolution_{i+1},\morphism_{i+1})$ is simple
however then we cannot assume that $\mysolution_{i} = \phi_{i+1}[\mysolution_{i+1}]$,
we can only guarantee that $\morphism_i(\mysolution_i(U_i)) = \morphism_{i+1}(\mysolution_{i+1}(U_{i+1}))$.
This makes details of many proofs more complicated, but this is a technical detail that should not bother the reader.

Getting back to the solutions, an equation in which both $U_i$ and $V_i$ have length $1$ has easy to describe solutions:
\begin{itemize}
	\item if $U_i, V_i$ are the same constant then the equation has exactly one solution,
	in which every variable is assigned $\epsilon$
	(recall our convention that a variable not present in the equation is assigned $\epsilon$);
	\item if $U_i, V_i$ are both variables, say $X$ and $Y$, then if $\rho(X) \neq \rho(Y)$ then there is not solution,
	otherwise any $\mysolution$ that assigns $\epsilon$ to other variables and $w$ to $X, Y$, where $\rho(w) = \rho(X)$,
	is a solution;
	\item if $U_i$ is a constant and $V_i$ a variable, say $a$ and $X$ then if $\rho(a) \neq \rho(X)$ then there is no solution,
	otherwise there is a unique solution, which assigns $a$ to $X$ and $\epsilon$ to all other variables.
\end{itemize}
In this way all solutions of the input equation $(U, V)$ are obtained by a path from $(U_0,V_0)$
to some satisfiable $(U_i,V_i)$ satisfying $|U_i| = |V_i| =1$ and
the solution of $(U,V)$ is a composition of operators from the families on the path applied
to the solution of $(U_i,V_i)$.
Note that there may be several ways to obtain the same solution, using different paths in the graph.

\subsection{Graph representation of all solutions}
\label{subsec: graph representation}
\subsubsection{Construction of the solution graph}
Using Lemma~\ref{lem: main} one can construct in \PSPACE{} a graph like representation of all solutions of a given word equation:
for the input equation $(U, V)$ we construct a directed graph $\mathcal{G}$ which has nodes labelled with proper equations.
Then for each strictly proper equation $(U_0,V_0)$ such that $|U_0| > 1$ or $|V_0| > 1$
we use Lemma~\ref{lem: main} to list all possible sequences for $(U_0,V_0)$.
For each such sequence $(U_0, V_0), (U_1, V_1), \ldots, (U_k, V_k)$ we put the edges
$(U_0,V_0) \to (U_1,V_1)$, $(U_1,V_1) \to (U_2,V_2)$, \ldots, $(U_{k-1},V_{k-1}) \to (U_k,V_k)$
and annotate the edges with the appropriate family of operators.
We lastly remove the nodes that are not reachable from the starting node and those that do not have a path to an ending node.

In this way we obtain a finite description of all solution of a word equation with regular constraints.
\begin{theorem}
There exists and can be effectively constructed a \PSPACE{} transducer that
given a word equation with regular constraints over a free monoid generates a finite graph representation of all its solutions.
\end{theorem}

Using Proposition~\ref{prop:dghftom05} we obtain a similar claim for word equation with rational constraints
over a free group.

\begin{corollary}\label{cor:maingroupthm}
There exists and can be effectively constructed a \PSPACE{} transducer that
given a word equation with regular constraints over a free group generates a finite graph representation of all its solutions.
\end{corollary}

\subsubsection{Families of operators}
\label{susubsec: operators}
Let us now describe the used family of operators.
Given an edge $(U, V) \to (U', V')$ the class $\Phi$ of operators is defined using:
\begin{itemize}
	\item A linear Diophantine system of polynomial size in parameters $\{x_X, y_X \}_{X \in \variables}$.
	\item A set $\{s_X, s_X' \}_{X \in \variables}$ of strings, length of string $s_X$ ($s_X'$) may depend on $x_X$
	($y_X$, respectively): it may use \emph{one} expression of the form $(ab)^{x_X}$ ($(ab)^{y_X}$, respectively)
	or $a^{x_X}$/$a^{y_X}$ when $a = b$.
	Each $s_X$ and $s_X'$ is of polynomial length (we treat $(ab)^{x_X}$ as having description of $\Ocomp(1)$ size).
	\item A set of $E_1, \ldots, E_k$ of strings which may use expressions $(ab)^{x_X}$ and $(ab)^{y_X}$
	(or $a^{x_X}$ and $a^{y_X}$), similarly to $s_X$ and $s'_X$,
	$k$ is of polynomial size and each $E_i$ has polynomial-size description.
	There are corresponding letters $c_{E_1}$, $c_{E_2}$, \ldots, $c_{E_k}$ that occur in $(U', V')$ but not in $(U, V)$.
\end{itemize}
Note that the Diophantine system may be empty and some of $\{s_X, s_X' \}_{X \in \variables}$ may be $\epsilon$
or not dependent on parameters.
On the other hand, each $E_i$ consists of at least two letters.

Particular operator $\phi \in \Phi$ corresponds to a solution $\{\ell_X, r_X \}_{X \in \variables}$.
It first replaces each letter $c_{E_i}$ with strings $E_i$ in which all $x_X$ and $y_X$ are replaced with numbers $\ell_X$ and $r_X$.
Then it prepends to \sol X the $s_X$ in which parameter $x_X$ is replaced with $\ell_X$, then appends $s_X'$
in which parameter $y_X$ is replaced with $r_X$.

\section{Compression step}
\label{sec: compression steps}

In this section we describe procedures that show the claim of Lemma~\ref{lem: main}.
In essence, for a word equation (with constraints) $(U,V)$ with a solution \mysolution{}
we want to \emph{compress} the word \sol{U} directly on the equation,
\ie without the knowledge of the actual solution.
In case of the free monoid (without involution)~\cite{wordequations},
the `compression' essentially is a replacement of all substrings $ab$ with a single constant $c$.
However, due to the involution (and possibility that $a = b$)
the compressions in case of free monoid with involution are more involved:
we replace the $ab$-\emph{blocks}, as defined later in this section, see Definition~\ref{definitionabblock}.
To do this, we sometimes need to modify the equation $(U,V)$.

The crucial observation is that a properly chosen sequence of such compression guarantees
that if the compressed equation is not too long than neither is the obtained equation
(formally, if the compressed equation is strictly proper then the resulting one is as well),
see Lemma~\ref{lem: size bound}.
Moreover, the compression steps decrease the weight of the corresponding solution,
which guarantees a termination of the whole process.

\subsection{Reducing the equation}

\subsubsection{Transforming solutions and inverse operators.}
As we want to describe the set of all solutions,
ideally there should be a one-to-one correspondence between the solutions before and after the application
of used subprocedures.
However, as those subprocedures are non-deterministic and the output depends on the non-deterministic choices,
the situation becomes a little more complicated.
What we want to guarantee is that no solution is `lost' in the process and no solution is `gained':
given a solution for some non-deterministic choices we transform the equation into another one, which has a `corresponding' solution
and we know a way to transform this solution back into the original equation.
Furthermore, when we transform back in this way any solution of the new equation, we obtain a solution of the original equation.

As already noticed, to ease the presentation, the solutions of the new equation may use constants outside the alphabet of the new equation.
To be more precise, they can `inherit' some constants from the previous solution and therefore use also the constants that occurred in the previous
equation.

\begin{definition}[Transforming the solution]
Given a (nondeterministic) procedure and a proper equation $(U, V)$ we say that this procedure
\emph{transforms} $(U, V)$ (which is proper) with its solution $(\mysolution,\morphism)$ 
to $(U', V')$ with $(\mysolution',\morphism')$
if 
\begin{itemize}
	\item
	there are some nondeterministic choices that lead to an equation $(U', V')$ (over the alphabet $\letters'$)
	and based on the nondeterministic choices and equation $(U, V)$ we can define a family of operators $\Phi$
	such that $\phi[\mysolution'] = \mysolution$
	for some solution $(\mysolution', \morphism')$ over the alphabet $\letters \cup \letters'$
	of the equation $(U', V')$ and some operator $\phi \in \Phi$.
	Furthermore, $\morphism(\sol U) = \morphism'(\mysolution'(U))$ and $\weight(\mysolution', \morphism') \leq \weight(\mysolution, \morphism)$
	and if $(U, V) \neq (U', V')$ then this inequality is in fact strict.
	\item For every equation $(U', V')$ that can be obtained from a proper equation $(U, V)$
	and any its solution $(\mysolution',\morphism')$
	(not necessarily simple) and for every operator $\phi \in \Phi$ the $(\phi[\mysolution'],\morphism)$
	is a solution of $(U, V)$ for any homomorphism $\morphism : B \mapsto A^+$ compatible with $\rho$,
	where $B$ is an alphabet of $\phi[\mysolution'](U)$. 
\end{itemize}
If this procedure transforms any solution $(\mysolution,\morphism)$ of any proper equation $(U, V)$
then we say that it \emph{transforms solutions}.
\end{definition}
Note that both $(U, V)$ and $\Phi$ depend on the nondeterministic choices, so it might
be that for different choices we can transform $(U, V)$ to $(U', V')$ (with a family $\Phi'$)
and to $(U'', V'')$ (with a family $\Phi''$).

We call $\Phi$ the \emph{corresponding family of inverse operators}.
In many cases, $\Phi$ consists of a single operator $\phi$,
in such a case we call it the \emph{corresponding inverse operator}
furthermore, in some cases $\phi$ does not depend on $(U, V)$.

Note that when $(U, V)$ with a $(\mysolution, \morphism)$ is transformed into $(U', V')$ with $(\mysolution', \morphism')$
then the simplification $\mysolution''$ of $\mysolution'$
(recall that a simplification replaces all constants $b \notin B'$ by $\morphism'(b)$ in all $\mysolution'(X)$)
is also a solution of $(U', V')$ and moreover $\morphism'(\mysolution''(U')) = \morphism(\sol U)$,
so it corresponds to the same original solution of the input equation as $(\sol U,\morphism)$,
see Lemma~\ref{lem: simplify the solution}.

Clearly, composition of two operations that (weakly) transform the equations also (weakly) transforms the equations
(although the description of the family of inverse operators may be more complex).

As a last comment, observe that when we take an arbitrary solution $(\mysolution', \morphism')$ and operator $\phi$
then we cannot guarantee that there is some $\morphism$ for which $\morphism(\phi[\mysolution'](U)) = \morphism(\mysolution'(U'))$:
imagine we can replace factor $ab\inv{a}$ with a single letter $c$ while $\morphism'(c) = a'b'$,
so there is no way to reasonably define $\morphism(a)$ and $\morphism(b)$. 
Thus we can take any $\morphism$ for $\phi[\mysolution']$, and we know that one exists by the assumption that $(U, V)$
is proper.

\subsubsection{$ab$-blocks.}
In an  earlier paper using the recompression technique~\cite{wordequations} there were two types of compression steps:
compression of pairs $ab$, where $a \neq b$ were two different constants,
and compression of maximal factor $a^\ell$ (\ie ones that cannot be extended to the right, nor left).
In both cases, such factors were replaced with a single fresh constant, say $c$.
While the actual replacement was performed only on the equation $(U,V)$
implicitly it was performed also on the solution $\sol U$ as well.

The advantage of such compression steps was that the replaced factors were non-overlapping,
in the sense that when we fixed a pair (or block) to be compressed,
each constant in a word $w$ belongs to at most one replaced factor.

We would like to use similar compression rules also for the case of monoids with involution,
however, one needs to take into the account that when $w$ is replaced with a constant $c$,
then also $\ov{w}$ should be replaced with $\ov{c}$.
The situation gets complicated, when some of constants in $w$ are fixed-points for the involution, \ie $\ov{a} = a$.
In the worst case, when $\ov{a} = a$ and $\ov{b} = b$ the occurrences of $ab$ and $\ov{ab} = ba$ are overlapping,
so the previous approach no longer directly applies. 
(Even if we start with a situation such that $a \neq \ov a$ for all  $a \in A$, as it is the case for free groups, fixed points 
in larger alphabets are produced during the algorithm.)

Intuitively, when we want to compress $ab$ into a single constant, also $\inv b \inv a$ needs to be replaced.
Furthermore, if factors $s$ and $s'$ are to be replaced and they are overlapping, we should replace their union with a single constant.
Lastly, the factors to be replaced naturally fall into \emph{types},
depending on whether the first constant of the factor is $a$ or $\inv b$ and the last $b$ or $\inv a$.

These intuitions lead to the following definition of $ab$-blocks (for a fixed pair of constants $ab$) and their types.
\begin{definition}\label{definitionabblock}
Depending on $a$ and $b$, $ab$-\emph{blocks} are
\begin{enumerate}
\item If $a=b$ then there are two types of $ab$-blocks: $a^i$ for $i\geq 2$ and $\inv a^i$ for $i\geq 2$.
\item If $a \neq b$, $\inv{a} \neq a$ and $\inv{b} \neq b$ then $ab$ and $\inv{ab} = \inv{b}\inv{a}$ are the two types of $ab$-blocks.
\item If $a \neq b$, $\inv{a}=a$ and $\inv{b} \ne b$ then $ab$, $\inv{ab} = \inv{b}a$ and $\inv{b}ab$ are the three types of $ab$-blocks.
\item If $a \neq b$, $\inv{a} \neq a$ and $\inv{b} = b$ then $ab$, $\inv{ab} = b\inv{a}$ and $ab\inv{a}$ are the three types of $ab$-blocks.
\item If $a \neq b$, $\inv{a} = a$ and $\inv{b} = b$ then the
$(ba)^i$, $a(ba)^i$, $(ba)^ib$ and $(ab)^i$ (where in each case $i \geq 1$)
 are the four types of $ab$-blocks.
\end{enumerate}
An occurrence of an $ab$-block in a word is an $ab$-\emph{factor}, it is \emph{maximal},
if it is not contained in any other $ab$-factor.
\end{definition}

Note that for the purpose of this definition when $a = \inv a$ we treat $a$ and $\inv a$ as the same letter,
even if for some syntactic reason we write $a$ and $\inv a$.

The following fact is a consequence of the definitions of maximal $ab$-blocks
and shows the correctness of the definition.
\begin{lemma}
\label{factoverlap}
For any word $w \in \letters^*$ and two constants $a, b\in \letters$, maximal $ab$-factors in $w$ do not overlap.
\end{lemma}
\begin{figure}
	\centering
		\includegraphics{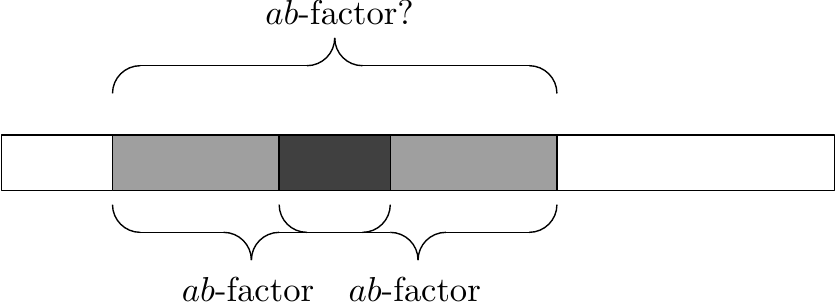}
\caption{To prove Lemma~\ref{factoverlap} we need to show that a union of two overlapping $ab$-factors is also an $ab$-factor.\label{fig:factors}}
\end{figure}

\begin{proof}
For the proof it is enough to show that when two $ab$-factors $s$ and $s'$ are overlapping then their union
(\ie the smallest factor that encompasses them both) is also an $ab$-factor,
see Figure~\ref{fig:factors}
Due to case distinction it follows by a case by case analysis according to Definition~\ref{definitionabblock}.

If $a = b$ and $\inv a \neq \inv a$ then as $ab$ and $\inv b \inv a$ have no common constants,
two overlapping $ab$-factors are both factors consisting of repetitions of the same constant
and so also their union is an $ab$-factor. If $\inv a = a$ then $ab = \inv b \inv a$ and so the same argument as before applies.

If $a \neq b$, $\inv{a} \neq a$ and $\inv{b} \neq b$ then two overlapping $ab$-factors need to be the same factor
(note that here we do not exclude the case $a = \inv b$).

If $a \neq b$, $\inv{a} = a$ and $\inv{b} \neq b$ then consider two different overlapping $ab$-factors.
As all constants in any $ab$-block are different, if $s$ and $s'$ are of the same type and overlapping then they are in fact the same factor.
If the factors $ab$ and $\inv b a$ overlap then their union is $\inv b a b$, which is also an $ab$-factor.
If factors $ab$ and $\inv b a b$ overlap, then the latter contains the former;
the same applies to the factors $\inv b a$ and $\inv b a b$.

If $a \neq b$, $\inv{a} \neq a$ and $\inv{b} = b$ then the analysis is symmetric to the one given above.

In the last case, when $a \neq b$, $\inv{a} = a$ and $\inv{b} = b$ observe that a factor is an $ab$-factor if and only if
it has length at least $2$ and consists solely of alternating constants $a$ and $b$.
Thus also a union of two overlapping $ab$-factors is an $ab$-factor.
\qed
\end{proof}

Given a set of $ab$-blocks we perform the compression by replacing maximal $ab$-factors from this set.
For consistency, we assume that such a set is involution closed.

\begin{definition}[$S$-reduction]
For a fixed $ab$ and an involution-closed set of $ab$-blocks $S$ the $S$-\emph{reduction} of the word $w$
is the word $w'$ in which all maximal factors $s \in S$ are replaced by a new constant $c_s$, where $\ov{c_s} = c_{\ov{s}}$.
The inverse operation is an $S$-\emph{expansion}.
\end{definition}

There are the following observations.
\begin{itemize}
	\item The $S$-expansion is a functions on $\letters^*$,
	using Lemma~\ref{factoverlap} we obtain that also $S$-reduction is a function on $\letters^*$ as well.
	\item The $S$-reduction introduces new constants to $\letters$, we extend $\rho$ to it in a natural way.
	\item We let $c = \ov c$ \IFF $s = \ov s$.
	In this way constants may become fixed point for the involution. 
	For example, $a \ov a$ is an $a \ov a$-block for $a \neq \ov a$.
	If  $a \ov a$ is compressed into $c$ then $c = \ov c$.
\item It might be that after $S$-reduction some constant $c$ in the solution is no longer in $\letters$
	(as it was removed from the equation).
	In such a case the corresponding solution will not be simple,
	this is described in more detail later on.
\end{itemize}

\subsubsection*{Performing the $S$-reduction}
The $S$-reduction is easy, if all maximal factors from  $S$ are wholly contained within the equation or within substitution for a~variable:
in such a case we perform the $S$-reduction separately on the equation and on each substitution for a variable (the latter is done implicitly).
It looks non-obvious, when part of some factor $s$ is within the substitution for the variable and part in the equation.
Let us formalise those notions.
\begin{definition}
For a word equation $(U,V)$ an $ab$-factor is {\em crossing} in a solution \mysolution{}
if it does not come from $U$ ($V$, respectively),
nor from any \sol X for an occurrence of a variable $X$;
$ab$ is \emph{crossing} in a solution \mysolution, if some $ab$-factor is crossing.
Otherwise $ab$ is \emph{non-crossing} in \mysolution.
\end{definition}
Note that as $\ov b \ov a$ is an $ab$-block, it might be that $ab$ is crossing because of a factor $\ov b \ov a$.

By guessing all $X\in \OO$ with $\sol X = \epsilon$ (and removing them) we can always assume that 
$\sol X \neq \epsilon$ for all $X$. In this case 
crossing $ab$'s can be alternatively characterized in a more operational manner.
\begin{lemma}
\label{lem: characterisation of crossing pairs}
Let $\sol X \neq \epsilon$ for all $X$.
Then $ab$ is crossing in \mysolution{} if and only if one of the following holds:
\begin{itemize}
	\item $aX$ or $\ov{aX} = \ov X \ov a$, for an unknown $X$, occurs in $U$ or $V$ and \sol X begins with $b$
	(so \sol{\ov X} ends with $\ov b$) \emph{or}
	\item $Xb$ or $\ov{Xb} = \ov{b} \, \ov{X}$, for an unknown $X$, occurs in $U$ or $V$ and \sol X ends with $a$
	(so \sol{\ov X} begins with $\ov a$) \emph{or}
	\item $XY$ or $\ov {XY} = \ov Y\, \ov X$, for unknowns $X, Y$, occurs in $U$ or $V$ and \sol X ends with $a$ while \sol Y 
	begins with $b$ (so \sol{\ov Y} ends with $\ov b$ and \sol{\ov X} begins with $\ov a$). 
\end{itemize}
\end{lemma}
\begin{proof}
So suppose that $ab$ is crossing, which means that there is some $ab$-factor that is crossing.
By definition it means that it does not come from one occurrence of a variable, nor from equation.
Thus one of its constants comes from a variable and the other from the equation or from a different variable
(note that as the solution is non-empty, the considered constants and variables are neighbouring in the word).
Case inspection implies that one of the conditions listed in the lemma describes this situation.

So suppose that $ab$ satisfies one of the conditions in the lemma. Then clearly $ab$ is crossing:
for instance, if $aX$ occurs in the equation and \sol X begins with $b$ then the $ab$-factor formed by this $a$
and the first $b$ in \sol X is crossing, other cases are shown in the same way.
\qed
\end{proof}

Since a crossing $ab$ can be associated with an occurrence of a variable $X$,
it follows that the number of crossing $ab$s is linear in the number of occurrences of variables.
\begin{lemma}
\label{lem: different crossing}
Let $(U,V)$ be a proper equation and \mysolution{} its solution.
Then there are at most $4n$ different crossing words in $\mysolution$.
\end{lemma}
The only property of a proper equation is that it has at most $n$ occurrences of variables.
\begin{proof}
Firstly observe that if $\sol X = \epsilon$ for any variable, then we can remove all $X$es from the equation
and this does not influence whether $ab$ is crossing or not.
Thus we can assume that \mysolution{} is non-empty, in the sense that $\sol X \neq \epsilon$ for every variable present in the equation.
So in the remaining part we may assume that the assumptions of Lemma~\ref{lem: characterisation of crossing pairs}
are satisfied.

By Lemma~\ref{lem: characterisation of crossing pairs} when $ab$ is crossing,
then one can associate $ab$ (or $\inv b \inv a$) with an occurrence of a variable and its first or last constant.
There are at most $n$ occurrences of variables, so $2n$ occurrences with distinguished first or last constant
and we have also two options of associating ($ab$ or $\inv b \inv a$), so $4n$ possibilities in total,
which yields the claim.
\qed
\end{proof}

\subsubsection*{Reduction for non-crossing $ab$.}
When $ab$ is non-crossing in the solution \mysolution{}
we can make the compression for all $ab$-factors that occur in the equation $(U,V)$
on \sol U by replacing each $ab$-factor in $U$ and $V$.
The correctness follows from the fact that each maximal occurrence of $ab$-block in \sol U and \sol V
comes either wholly from $U$ ($V$, respectively) or from \sol X.
The former are replaced by our procedure and the latter are replaced implicitly, by changing the solution.
Thus it can be shown that the solutions of the new and old equation are in one-to-one correspondence,
\ie are transformed by the procedure.

\begin{algorithm}[h]
  \caption{$\algncr(U,V, ab)$ Reduction for a non-crossing $ab$ \label{alg:pc}}
  \begin{algorithmic}[1]
		\State $S \gets $ all maximal $ab$-factors in $U$ and $V$
  	\For{$s \in S$}
			\State let $c_s$ be a fresh constant
			\If{$s = \ov{s}$}
				\State let $\ov{c_s}$ denote $c_s$
			\Else
				\State let $\ov{c_s}$ be a fresh constant
			\EndIf
			\State replace each maximal $ab$-factor $s$ ($\ov{s}$) in $U$ and $V$ by $c_s$ ($\ov{c_s}$, respectively)
			\State set $\rho(c_s) \gets \rho(s)$ and $\rho(\ov{c_s}) \gets \rho(\ov{s})$
		\EndFor
		\State \Return $(U', V')$
 \end{algorithmic}
\end{algorithm}

To show that $\algncr(U,V, ab)$ transforms the solutions (for a non-crossing $ab$)
or weakly transforms the solutions (in the general case),
for a solution $(\mysolution, \morphism)$ we should define a corresponding solution $(\mysolution', \morphism')$
of the obtained $(U', V')$ as well as an inverse operator $\phi$.
Intuitively, they are defined as follows (let as in $\algncr(U,V, ab)$
the $S$ be the set of all maximal $ab$-blocks in $(U, V)$
and let $\algncr(U,V,ab)$ replace $s \in S$ by $c_s$):
\begin{itemize}
	\item $\mysolution'(X)$ is obtained by replacing each $s \in S$ by $c_s$
	\item $\morphism'$ is $\morphism$ extended to new constants by setting
	$\morphism'(c_s) = \morphism(s)$
	\item the operator $\phi_{\{c_s \to s\}_{s \in S}}$ in each \sol X replaces each $c_s$ by $s$, for all $s \in S$.
(Note that $\phi_{\{c_s \to s\}_{s \in S}}$ is the $S$-expansion.)
\end{itemize}

Note that the defined operator is in the class defined in Section~\ref{subsec: graph representation}:
all $s_X, s_X'$ are $\epsilon$ while $E_1, \ldots, E_k$ are exactly the elements of $S$.

\begin{lemma}
\label{lem: compression noncrossing}
Let  $ab$ is non-crossing in a solution \mysolution{} and let $\algncr(U,V,ab)$ compute a set of $ab$-blocks $S$ in $(U, V)$ and replace $s \in S$ by $c_s$.
Then $\algncr(U,V,ab)$ 
transforms $(U,V)$ with $(\mysolution,\morphism)$ to $(U',V')$ with $(\mysolution',\morphism')$,
where $\morphism '$ is defined as above and 
$\phi_{\{c_s \to s\}_{s \in S}}$ is the inverse operator.
\end{lemma}
\begin{proof}
We define a new solution $\mysolution'$ by replacing each maximal factor $s \in S$ in any \sol X by $c_s$.
Note that in this way $\rho(\sol X) = \rho(\mysolution'(X))$, as for each $s$ we define $\rho(c_s) \gets \rho(s)$.
As all constants $\{c_s \to s\}_{s \in S}$ are fresh, this means that $\mysolution = \phi_{\{c_s \to s\}_{s \in S}}[\mysolution']$,
as claimed.
This is a solution of $(U', V')$: consider any maximal $ab$-factor $s$ in \sol U (or \sol V):
\begin{itemize}
	\item If it came from the equation then it was replaced by $\algncr(U,V,ab)$.
	\item If it came from a substitution for a variable and $s$
	\begin{itemize}
		\item is in $S$ then it was replaced implicitly in the definition of $\mysolution'$;
		\item is not in $S$ then it is left as it was.
	\end{itemize}
	\item It cannot be crossing, as this contradicts the assumption.
\end{itemize}
Thus, $\mysolution'(U')$ is obtained from \sol U by replacing each maximal $ab$-factor $s \in S$ by $c_s$,
in particular, $\mysolution'(U') = \mysolution'(V')$.

We define $\morphism'$ simply by extending $\morphism$ to new letter $c_s$ in a natural way:
$\morphism(c_s) = \morphism(s)$ for each $s \in S$;
note that such defined $\morphism'$ is compatible, since $\rho(c_s) = \rho(s)$ and $h$ is known to be compatible.
Furthermore $\morphism'(\mysolution'(U')) = \morphism(\mysolution(U))$,
as $\mysolution'(U')$ is obtained from $\mysolution(U)$ by replacing each maximal factor $s\in S$ by $c_s$
and by definition $\morphism'(c_s) = \morphism(s)$ and on all other letters they are equal.

Now, if $(\mysolution'', \morphism'')$ is any solution of $(U', V')$ then
$\phi_{\{c_s \to s\}_{s \in S}}[\mysolution''](U) = \phi_{\{c_s \to s\}_{s \in S}}[\mysolution''](V)$:
observe that $\phi_{\{c_s \to s\}_{s \in S}}[\mysolution''](U)$ is obtained from $\mysolution''(U')$ by replacing each $c_s$ by $s$,
as the same applies to $\phi_{\{c_s \to s\}_{s \in S}}[\mysolution''](V)$, we obtain that indeed $(\mysolution'',\morphism'')$
is a solution, for any $\morphism''$ for letters present in $\mysolution''(U)$.
We now show that there is at least one such a homomorphism:
we know that there is such a homomorphism for the alphabet of $(U, V)$ (as it is a proper equation)
and for other letters we can use the homomorphism $\morphism''$, by the form of $\phi_{\{c_s \to s\}_{s \in S}}$
there are no other letters.

Concerning the weight, observe first that $\morphism'(\mysolution'(X)) = \morphism(\mysolution(X))$:
\begin{itemize}
	\item if $c_s$ replaced $s$ then $\morphism'(c_s) = \morphism(s)$;
	\item for every preserved constant $a$ it holds that $\morphism'(a) = \morphism(a)$.
\end{itemize}
Clearly we have $|U'| + |V'| \leq |U| + |V|$, so the weight does not increase.
Furthermore, if $(U, V) \neq (U', V')$ then at least one factor was replaced in the equation
and so $|U'| + |V'| < |U| + |V|$ and so the weight decreases.
\qed
\end{proof}

\subsection{Reduction for crossing $ab$.}
Since we already know how to compress a non-crossing $ab$,
a natural way to deal with a crossing $ab$ is to `uncross' it and then compress using \algncr.
To this end we pop from the variables the whole parts of maximal $ab$-blocks which cause this block to be crossing.
Afterwards all maximal $ab$-blocks are noncrossing and so they can be compressed using $\algncr$

\subsubsection{Idea}
As an example consider an equation $abaXaXaXa = aXabYbYbY$, let $a = \inv a$ an $b = \inv b$ so that the $ab$-blocks are non-trivial.
For simplicity, let us for now ignore the constraints.
Also, let us focus on the solutions of the form $X \in b(ab)^{\ell_X}$ and $Y = (ab)^{\ell_Y}a$;
clearly, $ab$ is crossing in this solution.
So we `pop' from $X$ the $b(ab)^{\ell_X}$ and $(ab)^{\ell_Y}a$ from $Y$ (and remove those variables).
After the popping this equation is turned into $(ab)^{3 \ell_X+4}a = (ab)^{\ell_X + 3 \ell_Y + 4}a$, for which $ab$ is noncrossing.
Thus solutions of the original equation (of the predescribed form $X = b(ab)^{\ell_X}$ and $Y = (ab)^{\ell_Y}a$)
correspond to the solutions of the Diophantine equation:
$3 \ell_X + 4 = \ell_X  + 3 \ell_Y + 4$.
This points out another idea of the popping: when we pop the whole part of block that is crossing,
we do not immediately guess its length, instead we treat the length (here: $2\ell_X +1$ or $2 r_X +1$) as a \emph{parameter},
identify $ab$-blocks of the same length and only afterwards verify, whether our guesses were correct.
The verification is formalised as a linear system of Diophantine equations (here: $3 \ell_X + 4 = \ell_X  + 3 \ell_Y + 4$) in parameters (here: $\ell_X$ and $r_X$).
We can check solvability (and compute a minimal solution) in \NP{}
(so in particular in \PSPACE), see \eg \cite{HU}. (For or a more accurate estimation of constants see \cite{dgh05IC}).
Each of solutions of the Diophantine system corresponds to one ``real'' set of lengths of $ab$-blocks popped from variables.
Now we replace equation $(ab)^{3 \ell_X+4}a = (ab)^{\ell_X + 3 \ell_Y + 4}a$ with $ca = ca$, which has a unique solution
$\sol X = \sol Y = \epsilon$.
Of course there is no single inverse operator, instead, they should take into the account the system $3 \ell_X+4 = \ell_X + 3 \ell_Y + 4$.
And it is so, for each solution $(\ell_X, r_X)$ of this system there is one inverse operator,
which first replaces $c$ with $(ab)^{\ell_X + 3 \ell_Y + 4}$ and then appends $b(ab)^{\ell_X}$ to the substitution for $X$
and $(ab)^{r_X}a$ to the substitution for $Y$.

There are some details that were ignored in this example:
during popping we need to also guess the types of the popped blocks and whether the variable should be removed
(as now it represents $\epsilon$) or not.
Furthermore, we also need to calculate the transition of the popped $ab$-block,
which depends on the actual length (\ie on particular $\ell_X$, $\ell_Y$, etc.).
However, this $ab$ block is long because of repeated $ab$.
Now, when we look at $\rho(ab)$, $\rho(ab)^2$, \ldots then starting from some (at most exponential) value 
it becomes periodic, the period is also at most exponential.
Thus $\rho(ab)^{\ell_X} = \rho(ab)^{\ell}$ for some $\ell$ which is at most exponential.
This can be written as an Diophantine equation and added to the constructed linear Diophantine system
which has polynomial size if coefficients are written in binary.

\subsubsection{Detailed description}
A full description is available also as a psuedocode, see Algorithm~\ref{alg:prefixi}.
The proof of correctness is provided in Lemma~\ref{lem: compression crossing}.

\paragraph{Idempotent power}
In the preprocessing, when the $ab$-blocks can be nontrivial
(\ie when $a = b$ \emph{or} $a\neq b$ and $a = \inv a$ and $b = \inv b$)
we guess (some) idempotent power $p$ of $\rho(ab)$ (when $a \neq b$) or $\rho(a)$ (when $a = b$, in the following
we consider only the former case) in $\Mn$,
\ie $p$ such that $\rho(ab)^{2p} = \rho(ab)^p$.
It is easy to show that there is such $p \leq |\Mn| \leq 2^{4m^2}$,
so we can restrict the guess so that the binary notation of $p$ is of polynomial size.
Note that we can verify the guess by computing $\rho(ab)^{2p}$ and $\rho(ab)^p$ in time $\poly(\log p, m)$
(the powers are computed be iterated squaring of the matrices).
And so we can indeed verify that $p$ is an idempotent power of $\rho(ab)$.
Note that as $a =\inv a$ and $b = \inv b$ in this case, $p$ is also an idempotent power of $ba = \inv{ab}$.

\paragraph{Popping and transitions}
Now for every variable $X$ we guess whether \sol X begins (and ends) with an $ab$-factor
or a single-letter suffix (prefix, respectively) of an $ab$-factor;
to simplify the notation, the $ab$-\emph{prefix} of \sol X is the maximal prefix of \sol X that is also a suffix of some $ab$-factor;
define the $ab$-\emph{suffix} in a symmetric way.
Note that an $ab$prefix of \sol X may be empty, may consist of a single-letter (\ie $a$ or $\inv b$) or have more letters
in which case it is also an $ab$-factor itself.
Thus, for each $X$ we guess its $ab$-prefix $s_X$ and left-pop it from $X$, \ie we replace $X$ with $s_XX$
(at the same time we need to right-pop the $\inv{s_X}$ from $\inv X$, \ie replace $\inv X$ with $\inv X \inv{s_X}$,
note that this is the $ab$-suffix of $\inv X$).
Consider $s_X$, suppose that it is nontrivial, \ie more than $1$ letter was popped.
When $a = b$ then $s_X = a^{x_X}$ or $s_X = \inv{a}^{x_X}$ for some $x_X \geq 2$.
Similarly, when $a = \inv a, b = \inv b$ and $a \neq b$ then $s_X \in \{ (ab)^{x_X}, (ab)^{x_X}b, b(ab)^{x_X}, (ab)^{x_X} \}$,
for some $x_X \geq 1$.
In other cases, $s_X$ does not include any parameter $x_X$.
Similar observation can be made for $s'_X$, which uses parameter $y_X$.
We treat $x_X, y_X$ as \emph{parameters} denoting integers whose values are to be established later on
(we do not use name \emph{variables}, as this is reserved for variables representing words).
The information about the $x_X, y_X$ is encoded in $s_X, s'_X$: we simply write them in one of the forms given above
(note that $x_{\inv X} = y_X$ and $y_{\inv X} = x_X$).
Eventually, we fix the values of $x_X$ and $y_X$, say to $\ell_X$ and $r_X$.
Then $s_X[\ell_X]$ denotes a string $s_X$ in which we substituted a number $\ell_X$ for parameter $x_X$ and so $s_X[\ell_X]$
is a string of a well-defined length.
If $s_X$ does not depend on $\ell_X$ then also $s_X[\ell_X]$ is a string that does not depend on $\ell_X$,
still we use this notion to streamline the presentation.

We now fix the transitions of the popped blocks; if $x_X$ and $y_X$ are defined, the transition depends on them.
As an example, consider a transition $\rho(ab)^{x_X}$.
Recall that $p$ is the idempotent power for $\rho(ab)$ (and so also of $\rho(\inv b \inv a)$),
which means that $\rho(ab)^{kp + \ell'} = \rho(ab)^{p + \ell'}$ (when $k \geq 1$), where $0 \leq \ell' < p$.
Thus for $\rho(ab) ^{x_X}$ we guess whether $x_X \geq 2p$.
If so, we guess $0 \leq \ell'_X < p$ and write equations $kp + \ell'_X = x_X$, $k \geq 1$,
then $\rho(ab)^{x_X} = \rho(ab)^{p+\ell'_X}$, and both $p$ and $\ell'_X$ are known,
so we can calculate $\rho(ab)^{p + \ell'}$ .
If $x_X \leq 2p$ then we guess its value (at most $2p$) and compute $\rho(ab)^{x_X}$.
The same is done for $y_X$.

\paragraph{Identical blocks}
As we know the types of $ab$-factors in the equation (note that they do not depend on particular values of $\{x_X, y_X \}_{X \in \variables}$),
we can calculate the maximal $ab$-factors in the equation (as well as their types), even though $x_X$ and $y_X$ are not yet known.
Denote those maximal $ab$-factors by $E_1, \ldots, E_\ell$, note that they may use $(ab)^{x_X}$ or $(ab)^{y_X}$,
similarly as $\{s_X, s_X'\}_{X \in \variables}$,
we do not impose the condition that one $E_i$ uses at most one such an expression.
As in case of $\{s_X, s_X'\{_{X \in \variables}$, by $E_i[\{\ell_X, r_X \}_{X \in \variables}]$ we denote $E_i$
in which parameters $x_X$ and $y_X$ were replaced with numbers $\ell_X$ and $r_X$, for each variable $X$.
Concerning their lengths, denote by $e_i$ the length fo $E_i$.
since the popped factors have lengths that are linear in $x_X$ or $y_X$ for some $X$
the lengths $e_1, e_2, \ldots, e_\ell$ are also linear in $\{x_X, y_X \}_{X \in \variables}$.
It is easy to see (Lemma~\ref{lem: compression crossing}) that $\ell$ is polynomial in the size of the equation
and so are the descriptions of each $e_i$.
For the future reference, by $e_i[\{\ell_X, r_X \}_{X \in \variables}]$
we denote the evaluation of expression $e_i$ when $x_X$ is substituted by $\ell_X$
and $y_X$ is substituted by $r_X$, for each $X \in \variables$.
This corresponds to the length for some particular values of parameters $\{x_X, y_X \}_{X \in \variables}$.

Consider all maximal $ab$-factors of the same type, we guess the order between their lengths,
\ie we guess which of them are equal and what is the order between groups of expressions denoting equal lengths.
We write the corresponding conditions into the system of equations;
formally we divide these expressions into groups $\mathcal E _1$, $\mathcal E _2$, \ldots, $\mathcal E _k$,
one group contains only factors of the same type and its elements correspond to factors of the same length.
For each group $\mathcal E = \{E_{i_1}, E_{i_2}, \ldots, E_{i_\ell} \} $ we add equations
$e_{i_1} = e_{i_2}$, $e_{i_2} = e_{i_3}$, \ldots, $e_{i_{\ell-1}} = e_{i_\ell}$, which ensure that indeed those factors are of the same length.
Then for all groups  $\mathcal E _1$, $\mathcal E _2$, \ldots, $\mathcal E _k$ of factors of the same type we guess
the relation between the lengths of factors between the groups, \ie. for each two groups $\mathcal E_i$ and $\mathcal E_j$
we choose elements from the group, say $E_i$ and $E_j$, and add the appropriate of the inequalities $e_i < e_j$ or $e_i > e_j$ to the system.
Note that we can rule out the possibility that $e_i = e_j$, as we put $e_i$ and $e_j$ in different groups

\paragraph{Verification}
As the constructed Diophantine system $D$ is of polynomial size, it can be non-deterministically verified in polynomial time,
in particular, this can be done in \PSPACE, see \eg \cite{HU}.
If the verification fails, we terminate.

\paragraph{Replacement}
When the system $D$ is successfully verified, we replace all $ab$-factors in one group by a new letter,
\ie blocks $\{E_{i_1}, E_{i_2}, \ldots, E_{i_k} \}$ are replaced with a letter$a_{e_{i_1}}$
(the choice of $e_{i_1}$ is arbitrary), obtaining the new equation.

\begin{algorithm}[h]
  \caption{$\algcr(U,V,ab)$ Compression of $ab$-blocks for a crossing $ab$ \label{alg:prefixi}}
  \begin{algorithmic}[1]
		\State $p \gets$ idempotent power of $\rho(ab)$ in $\Mn$
		\Comment{Guess and verify when needed. The same as for $\inv{ab}$}
		\For{$\{X, \inv X\}\in \variables$} \Comment{Consider $X$ and its involution at the same time}
			\State guess $ab$-prefix of $s_X$ of \sol X \Comment{May depend on a parameter $x_X$}
				\State add constraint $x_X \geq 1$ (or $x_X \geq 2$), when applicable \Comment{Depending on $s_X$ and whether $a = b$}\par
				\Comment{$x_X \geq 1$ for $a\neq b$, $\inv a = a$, $\inv b = b$ and $x_X \geq 2$ for $a = b$}
				\If{$x_X < 2p$} \Comment{Guess when applicable}
					\State guess $\ell_X$, where $\ell_X < 2p$ \Comment{value of $x_X$}
					\State add $x_X = \ell_X$ to $D$, calculate $\rho_x \gets \rho(s_X[\ell_X])$
				\Else
					\State guess $\ell_X'$, where $0 \leq \ell_X' < p$ \Comment{value of $x_X$ mod $p$}
					\State add $\{x_X = k \cdot p + \ell_X', k > 0\}$ to $D$, calculate $\rho_s \gets \rho(s_X[\ell_X' + p])$
				\EndIf
				\State guess $\rho_X$ such that $\rho(X) = \rho_s \rho_X$
				\State replace each $X$ with $s_XX$, set $\rho(X) \gets \rho_X$
				\If{$\sol X = \epsilon$ and $\rho(X) = \rho(\epsilon)$} \Comment{Guess}
					\State remove $X$ from the equation
				\EndIf
			\State Perform symmetric actions on the end of $X$ \Comment{With parameter $y_X$}
		\EndFor
		\State let $E_1, E_2, \ldots, E_\ell$ be the maximal $ab$-factors in $(U,V)$ and
		$e_1, e_2, \ldots, e_\ell$ their lengths
  	\State partition $\{E_1, E_2, \ldots, E_\ell\}$ into groups $\{\mathcal E_1, \ldots, \mathcal E_k\}$, 
		\Comment{Guess the partition}
		\par
		\Comment{Each group has $ab$-factors of the same type}
  		
  	\For{each group $\mathcal E_{i_j} = \{ E_{i_1}, E_{i_2}, \ldots, E_{i_\ell} \}$}
				\State add equations $\{e_{i_1} = e_{i_2}, e_{i_1} = e_{i_2}, \ldots, e_{i_{\ell-1}} = e_{i_\ell}\}$ to $D$
  	\EndFor
  	\For{different groups $\mathcal E_{i}$ and $\mathcal E_{j}$ of $ab$-factors of the same type}
				\State take any $e_i \in \mathcal E_i$, $e_j \in \mathcal E_j$, add one of inequalities
				$\{e_i < e_j\}$ or $\{e_i < e_j\}$ to $D$
  	\EndFor
		\State verify system $D$ \Comment{In NP}
		\For{each part $\mathcal E_i = \{E_{i_1}, E_{i_2}, \ldots, E_{i_{\ell}}\}$}
			\State let $c_{e_{i_1}}$ be an unused constant
			\State replace blocks $E_{i_1}, E_{i_2}, \ldots, E_{i_{\ell}}$ by $c_{e_{i_1}}$
		\EndFor
  \end{algorithmic}
\end{algorithm}

\paragraph{Family of inverse operators}
The corresponding family of inverse operators is defined in terms of system $D$,
the popped prefixes and suffixes $\{s_X, s'_X \}_{X \in \variables}$
and the maximal blocks $E_1, E_2, \ldots, E_k$ (replaced with letters $c_{e_1}, c_{e_2}, \ldots, c_{e_k}$),
call this class $\Phi_{D, \{s_X, s'_X \}_{X \in \variables}, E_1, \ldots, E_k}$.
For each solution $\{\ell_X, r_X \}_{X \in \variables}$ of $D$ the family
contains an operator $\phi_{\{\ell_X,r_X\}_{X \in \variables}}$.
The action of such an operator (on $X$) are as follows: it first replaces each 
letter $c_{e_i}$ with $ab$-block $E_i[\{\ell_X,r_X\}_{X \in \variables}]$
(so of length $e_i[\{\ell_X,r_X\}_{X \in \variables}]$).
Afterwards, we append/prepend blocks $s_X[\ell_X]$ and $s'_X[r_X]$ to the substitution for $X$.

Note that this family of operators is of the form promised in Section~\ref{susubsec: operators}.

\paragraph{Inverse operator for a particular solution}
For a solution $(\mysolution, \morphism)$ consider the run of $\algcr(U,V,ab)$
in which the non-deterministic choices are done according to \mysolution:
\ie we guess $\{s_X, s'_X\}_{X \in \variables}$ such that $s_X[\ell_X]$ is the $ab$-prefix of $\sol X$
and $s_X'[r_X]$ is the $ab$-suffix of \sol X for appropriate values $\{ \ell_X, r_X\}_{X \in \variables}$
(when $\sol X = s_X[\ell_X]$ we guess $s_X' = \epsilon$).
We partition the arithmetic expressions according to \mysolution,
\ie $E_i$ and $E_j$ (of the same type) are in one group if and only if
$E_i[\{\ell_X,r_X\}_{X \in \variables}] = E_j[\{\ell_X,r_X\}_{X \in \variables}]$
(which in particular implies that their lengths
$e_i[\{\ell_X,r_X\}_{X \in \variables}]$ and $E_j[\{\ell_X,r_X\}_{X \in \variables}]$ are equal).
Additionally, for two different groups $\mathcal E _i$ and $\mathcal E_j$ of expressions of the same type
we add an equation $e_i < e_j$ for some $e_i \in \mathcal E_i$ and $e_j \in \mathcal E_j$ if and only if 
$e_i[\{\ell_X,r_X\}_{X \in \variables}] < e_j[\{\ell_X,r_X\}_{X \in \variables}]$
(and in the other case we add the converse inequality $e_i > e_j$).

Then $\{\ell_X,r_X\}_{X \in \variables}$ is a solution of a constructed system $D$ and
$\algcr(U,V,ab)$ transforms $(U, V)$ with $(\mysolution, \morphism)$ and 
$\phi_{\{\ell_X,r_X\}_{X \in \variables}} \in \Phi_{D, \{s_X, s'_X \}_{X \in \variables}, E_1, \ldots, E_k}$
is the corresponding inverse operator.
Concerning homomorphism $\morphism'$, we extend $\morphism$ to new constants by setting
$\morphism'(c_{e_i}) = \morphism (e_i[\{\ell_X,r_X\}_{X \in \variables}])$.

It remains to formally state and prove the above intuitions.

\begin{lemma}
\label{lem: compression crossing}
$\algcr(U,V,ab)$ transforms solutions.
Let $D$ be the system returned by $\algcr(U,V,ab)$ for the corresponding non-deterministic choices,
and let $X$ left-popped $s_X$ and right-popped $s'_X$
and let expressions in groups $\mathcal E_1, \ldots, \mathcal E_k$ be replaced with letters $c_{e_1}$, \ldots, $c_{e_k}$.
Then the family $\Phi_{D, \{s_X,s'_X\}_{X \in \variables}, E_1, \ldots, E_k}$ is the corresponding family of operators.
\end{lemma}
\begin{proof}
Let us focus on a proper equation $(U, V)$.
As a first step, we shall show that indeed all maximal blocks have lengths that are arithmetic expressions in
$\{x_X, y_X \}_{X \in \variables}$, there are polynomially (in $|U| + |V|$) many such lengths
and that each of them of them is also of polynomial size.
Consider, what letters can be included in a maximal $ab$-factor.
As the equation is proper, before any popping there are $\Ocomp(|U| + |V|)$ letters in the equation.
There are at most $2(|U| + |V|)$ popped factors (two for each occurrence of a variable)
and each of the length is at most $3 + 2x_X$ (or $3 + 2y_X$).
Hence, the total sum of lengths is $\Ocomp(|U| + |V|)$ plus $2 \sum_{X \in \variables}  (x_X + y_X)$,
as claimed.
Now, every nonempty popped $s_X$ (and $s'_X$) goes into exactly one maximal $ab$-factor, so indeed the lengths
are expressions linear in $\{x_X,y_X\}_{X \in \variables}$.

We now show that if $(U, V)$ has a solution $(\mysolution, \morphism)$ then for appropriate non-deterministic
choices we transform it into $(U', V')$ with $(\mysolution', \morphism')$ and the inverse operator is in the defined family.
So consider such a solution.
As already noted, consider the non-deterministic guesses of $\algcr(U,V,ab)$ that are consistent with $(\mysolution, \morphism)$,
\ie for each variable $X$ let its $ab$-prefix and $ab$-suffix be $s_X[\ell_X]$ and $s'_X[r_X]$
(note that it may be that $s_X$ does not depend on the parameter $x_X$, or $s'_X$ on $y_X$,
it may be that one of them is $\epsilon$; additionally, when $\sol X = s_X[\ell_X]$, we take $s_X' = \epsilon$).
Let $\algcr(U,V,ab)$ guess those $s_X$ and $s'_X$.
Let also $\algcr(U,V,ab)$ remove $X$ from the equation only when this is needed,
\ie $\sol X = s_X[\ell_X]s'_X[r_X]$.

Consider the equation $(U_1, V_1)$ obtained from the equation calculated so far by $\algcr(U,V,ab)$ 
by substituting $\{\ell_X, r_X\}_{X \in \variables}$ for $\{x_X, y_X\}_{X \in \variables}$.
Then it has a solution $(\mysolution_1, \morphism)$, where $\sol X = s_X[\ell_X] \mysolution_1(X) s'_X[r_X]$.
Moreover, $\sol U = \mysolution_1(U_1)$.
This is easy to see: $s_X[\ell_X]$ and $s'_X[r_X]$ are the $ab$-prefix and $ab$-suffix of $\sol X$ (by their definition)
and we replace $X$ by $s_X[\ell_X] X s'_X[r_X]$ (or $s_X[\ell_X] s'_X[r_X]$).

Let $E_1, \ldots, E_\ell$ be the maximal $ab$-factors calculated by the $\algcr(U,V,ab)$.
Then in $(U_1, V_1)$ the maximal $ab$-factors are
$E_1[\{\ell_X, r_X\}_{X \in \variables}], \ldots, E_\ell[\{\ell_X, r_X\}_{X \in \variables}]$ and have lengths
$e_1[\{\ell_X, r_X\}_{X \in \variables}]$, \ldots, $e_\ell[\{\ell_X, r_X\}_{X \in \variables}]$:
the $s_X$ and $s'_X$ were chosen so that they are of the type of the $ab$-prefix and $ab$-suffix of \sol X
and $s_X[\ell_X]$ and $s'_X[r_X]$ are the prefix and suffix of $\sol X$.

Lastly, the $ab$ is non-crossing in $(U_1, V_1)$ in $\mysolution_1$: suppose that it is not.
As we assumed that we removed $X$ when $\sol X = s_X[\ell_X]s'_X[r_X]$ then this means that $\mysolution_1$
is non-empty and so we can apply Lemma~\ref{lem: characterisation of crossing pairs}.
As the cases listed in the lemma are symmetric, suppose that $aX$ occurs in $(U_1, V_1)$ and $\mysolution_1(X)$
begins with $b$.
If $s_X = \epsilon$ then this is a contradiction, as then $\sol X$ also begins with $b$ and so we guessed 
the $ab$-prefix of \sol X incorrectly.
Thus $s_X \neq \epsilon$ and so also $s_X[\ell_X] \neq \epsilon$.
If $s_X[\ell_X]$ consists of at least two letters then it is an $ab$-factor and it overlaps with the $ab$-factor
consisting of the last letter of $s_X[\ell_X]$ and the following $b$ and so by Lemma~\ref{factoverlap}
the $s_X[\ell_X]b$ is also an $ab$-factor, which contradicts the choice of $s_X$.
If $s_X$ is a single letter, \ie $a$, then we clearly guessed incorrectly:
$s_X[\ell_X]\mysolution_1(X)$ begins with $ab$ and so also \sol X begins with $ab$,
thus we should have popped an $ab$-factor from it and not a single $a$.
The other cases are shown in the same way.

At this moment $\algcr(U,V,ab)$ also calculates the transition $\rho_s$ as well as adds some equations to the system.
Let it make the following choices:
if $\ell_X < 2p$ then let it guess that $x_X < 2p$ and guess $\ell_X$ as the value for $x_X$.
Then $\ell_X$ satisfies the added equations $x_X < 2p$ and $x_X = \ell_X$.
We calculate the transition $\rho_s = \rho(s_X[\ell_X])$ and there is a transition $\rho_X = \rho(\mysolution_{1}(X))$
such that $\rho_s \rho_X = \rho(X)$;
we make the corresponding nondeterministic choices.
If $\ell_X \geq 2p$ then let $\algcr(U,V,ab)$ guess this. The added inequality $x_X \geq 2p$ is satisfied by $\ell_X$.
Additionally, we guess $\ell'_X = \ell_X \mod p$, then the added equations $x_X = k p + \ell'_X$ and $k \geq 1$ are satisfiable by $\ell_X$ and some $k$ (which is irrelevant later on).
Moreover, as $p$ is an idempotent power for $\rho(ab)$ (or $\rho(a)$, when $a = b$;
for the simplicity of presentation in the following we consider only the former case),
we know that $\rho(s_X[\ell_X]) = \rho(s_X[\ell'_X + p])$, so $\rho_s = \rho(s_x[\ell_X])$ and it is also correctly calculated.
Hence we can guess $\rho_X$ to be the transition for $\rho(\mysolution_1(X))$ and then
$\rho(X) = \rho_s \rho_X$.

Consider the equations and inequalities on $e_1, \ldots, e_k$ added to $D$.
An equality $e_i = e_j$ is added if and only if $E_i[\{\ell_X, r_X\}_{X \in \variables}] = E_j[\{\ell_X, r_X\}_{X \in \variables}]$
(which implies that $e_i[\{\ell_X, r_X\}_{X \in \variables}] = e_j[\{\ell_X, r_X\}_{X \in \variables}]$)
and we consider the choices in which inequality $e_i < e_j$
is added only when the corresponding blocks are of the same type and
$e_i[\{\ell_X, r_X\}_{X \in \variables}] < e_j[\{\ell_X, r_X\}_{X \in \variables}]$.
Thus $\{\ell_X, r_X\}_{X \in \variables}$ satisfy those equations and inequalities as well.

Let us now investigate the replacement of $ab$-factors by $\algcr(U,V,ab)$.
Recall that we take the non-deterministic choices in which it assigns $E_i$ and $E_j$
into the same group if and only if $E_i[\{\ell_X, r_X\}_{X \in \variables}] = E_j[\{\ell_X, r_X\}_{X \in \variables}]$
and they represent factors of the same type. Then the corresponding $ab$-factors in $(U_1, V_1)$ are equal.
Thus the action of $\algcr(U,V,ab)$ are equivalent to $\algncr(U_1,V_1,ab)$ (up to naming of the new letters),
recall that we already shown that $ab$ is non-crossing in $\mysolution_1$.
Lemma~\ref{lem: compression noncrossing} guarantees that when $ab$ is non-crossing in $\mysolution$
then $\algncr(U_1,V_1,ab)$ transforms the solution $\mysolution$ and the inverse operator replaces letters
$c_{e_i}$ with the corresponding blocks of length
$e_i[\{\ell_X, r_X\}_{X \in \variables}]$.
So let $(U_1, V_1)$ with $(\mysolution_1, \morphism_1)$ be transformed to $(U', V')$ with $(\mysolution', \morphism')$,
as guaranteed Lemma~\ref{lem: compression noncrossing}.
By the same lemma we know what is the inverse operator that transforms $(U', V')$ with $(\mysolution', \morphism')$
to $(U_1, V_1)$ with $(\mysolution_1, \morphism_1)$.
From previous considerations we also know what is the inverse operator that transforms $(U_1, V_1)$ with $(\mysolution_1, \morphism_1)$
to $(U, V)$ with $(\mysolution, \morphism)$.
It is easy to see that their composition is exactly
$\phi_{\{\ell_X, r_X\}_{X \in \variables}}$ from $\Phi_{D, \{s_X, s'_X\}_{X \in \variables}, E_1, \ldots, E_k}$.
As $\{\ell_X, r_X\}_{X \in \variables}$ is a solution of $D$, this shows the the appropriate inverse operator indeed is in
$\Phi_{D, \{s_X, s'_X \}_{X \in \variables}, E_1, \ldots, E_k}$.

Concerning the weight, note that Lemma~\ref{lem: compression noncrossing}
shows that $\weight(\mysolution',\morphism') \leq \weight(\mysolution_1,\morphism_1)$
and the inequality is strict if $(U', V') \neq (U_1, V_1)$.
Similarly, since $\mysolution_1(X) = s_X[\ell_X]\sol X s'_X[r_X]$ and each $X$ was replaced with $s_X[\ell_X] X s'_X[r_X]$,
each popped $s_X[\ell_X]$ introduces $|s_X[\ell_X]|$ to $\weight(\mysolution',\morphism')$,
while in $(\mysolution_1, \morphism_1)$ it introduced $2 |\morphism_1(s_X[\ell_X])|$,
the same applies to $s_X'[r_X]$.
Thus $\weight(\mysolution_1,\morphism_1) \leq \weight(\mysolution,\morphism)$ and if any of $s_X, s_X'$ is non-empty,
the inequality is strict.
In the end, $\weight(\mysolution',\morphism') \leq \weight(\mysolution,\morphism)$
and the equality happens only when no factor was replaced and nothing was popped, \ie when $(U, V) = (U', V')$, as claimed.

We now move to the next part of the proof.
Assume that $(U, V)$ is turned into the equation $(U', V')$ that has a solution $(\mysolution', \morphism')$
and system $D$ was created on the way;
let also $s_X$ and $s'_X$ be popped to the left and right from $X$ (any of those may be $\epsilon$),
finally, let blocks from partition parts $\mathcal E_1$, $\mathcal E_2$, \ldots, $\mathcal E_k$ be replaced with letters $c_{e_1}$, $c_{e_2}$, \ldots, $c_{e_k}$.
We are to show that for any operator $\phi \in \Phi_{D, \{s_X, s'_X \}_{X \in \variables}, E_1, \ldots, E_k}$
the $(\phi[\mysolution], \morphism )$ is a solution of $(U, V)$, for any homomorphism $\morphism$ for $\phi[\mysolution](U)$
compatible with $\rho$ (and that there is such a homomorphism $\morphism$).

First observe that $\phi$ corresponds to some solution $\{\ell_X,r_X\}_{X \in \variables}$ of $D$.

Consider the equation obtained by $\algcr(U,V,ab)$ after popping letters but before replacement of $ab$-factors,
\ie the one using parameters $\{ x_X, y_X\}_{X \in \variables}$.
Consider the equation obtained by substituting $\{ \ell_X, r_X\}_{X \in \variables}$ for those parameters,
\ie each $s_X$ is replaced with $s_X[\ell_X]$ and each $s'_X$ by $s'_X[r_X]$.
Denote this equation by $(U_1, V_1)$.
If $(\mysolution_1, \morphism)$ is a solution of $(U_1, V_1)$ then $(\mysolution, \morphism)$
is a solution of $(U, V)$, where $\sol X$ is obtained from $\mysolution(X)$
by appending $s_X[\ell_X]$ and prepending $s'_X[r_X]$ to $\mysolution_1(X)$:
 \begin{itemize}
	 \item Since $\sol X = s_X[\ell_X]\mysolution_1(X) s'_X(r_X)$ and $(U_1, V_1)$ was obtained by replacing $X$ with $s_X[\ell_X] X s'_X[r_X]$
	(or $s_X[\ell_X] s'_X[r_X]$ and then $\mysolution_1(X) = \epsilon$), we get that $\sol U = \mysolution_1(U_1)$ and similarly 
	$\sol V = \mysolution_1(V_1)$.
	\item For the constraints: $\rho(\mysolution_1(X_1)) = \rho_X$ calculated by $\algcr(U,V,ab)$
	and satisfying the condition $\rho(X) = \rho_{s_X} \rho_X \rho_{s'_X}$.
	As $\{\ell_X, r_X \}$ is a solution of $D$ then $\algcr(U,V,ab)$ correctly calculated
	$\rho_{s_X} = \rho(s_X[\ell_X])$ and $\rho_{s_X'} = \rho(s'_X[r_X])$.
	\item For the involution, note that we assume that $\inv{\mysolution_1(X)} = \mysolution_1(\inv X)$
	and $\inv {s_X} = s'_{\inv X}$ and so we get that $\inv{\sol X} = \sol{\inv X}$.
 \end{itemize}

Note that the inverse operator transforming the solutions of $(U_1, V_1)$ to solutions of $(U, V)$
is a `second part' of the inverse operator $\phi$.

Now consider an equation obtained from $(U', V')$ by replacing every letter $a_{e_i}$ by an $ab$-factor
$E_i[\{\ell_X,r_X\}_{X \in \variables}]$.
Change the solution of $\mysolution'$ in the same way, obtaining $\mysolution_1$.
Concerning the constraints, note that we define $\rho(c_{e_i})$ so that $\rho(c_{e_i}) = \rho(e_i[\{ \ell_X, r_X\}_{X \in \variables}])$,
assuming that we properly calculate the latter.
As $\{ \ell_X, r_X\}_{X \in \variables}$ is a solution of $D$ and we calculated the idempotent power $p$ correctly,
this is the case: when $\ell_X < 2p$ we make the calculations explicit and otherwise 
we have that $\rho(ab)^{\ell_X} = \rho(ab)^{p + \ell_X'}$, where $\ell_X \mod p = \ell_X'$, which is calculated by the algorithm.
Note that this change is the first part of the operation performed by $\phi$.

Observe that $\mysolution_1(U_1) = \mysolution_1(V_1)$, as each letter $a_{e_i}$ was replaced in the same way in the
equation and in the solution, so $\mysolution_1$ is a solution of $(U_1, V_1)$, as claimed.
As equation is proper, we take any homomorphism $\morphism_1$,
so it is left to show that at least one such a homomorphism exists:
for letters that are present in $(U, V)$ note that a compatible $h$ can be defined for them, as we assumed that $(U, V)$ is proper.
Letters in \sol U that are not in $(U, V)$ were taken from $(U' ,V')$ and so a compatible $h'$ for such letters is known to exist.
\qed
\end{proof}

\subsubsection*{Main transformation.}
The main procedure $\algonephase(U, V)$ first lists all $ab$s that are either factors in $(U, V)$ or are crossing in \mysolution{}
(note that the latter need to be guessed).
While any of them is non-crossing and present in the equation, we compress this factor (and remove it from the list).
When none factor in the list is non-crossing, we guess the crossing $ab$s
(note that we always include all the remaining factors in the list).
Then for each of those factors we compress it using \algcr. 
\begin{algorithm}[H]
  \caption{$\algonephase(U, V)$ \label{alg: onepahsepopall}}
  \begin{algorithmic}[1]
		\State $P \gets $ list of explicit or crossing $ab$'s in $U$, $V$ \Comment{At most $|U| + |V| + 4n$}
		\While{there is a non-crossing $ab \in P$ such that $ab$ is a factor in $U$ or $V$}
				\State $\algncr(U, V, ab)$ 
				\State remove $ab$ from $P$
		\EndWhile
		\State $P' \gets $ crossing $ab$'s \Comment{Done by guessing first and last constants of each \sol X, $|P'| \leq 4n$}
		\par
		\Comment{$P'$ contains all factors from $P$ still occurring in the equation}
		\For{$ab \in P'$}
			\State $\algcr(U, V,ab)$
		\EndFor
		\State \Return $(U, V)$
 \end{algorithmic}
\end{algorithm}

The crucial property of \algonephase{} is that it uses equation of bounded size,
as stated in the following lemma. Note that this bound \emph{does not} depend on the non-deterministic choices of \algonephase.

\begin{lemma}
	\label{lem: size bound}
Suppose that $(U,V)$ is a strictly proper equation.
Then during \algonephase{} the $(U,V)$ is a proper equation and after it is strictly proper.
\end{lemma}
\begin{proof}
Consider, how many constants are popped into the equation during \algonephase.
For a fixed $ab$, \algcr{} may introduce long $ab$-blocks at sides of each variable, but then they are immediately replaced with one constant,
so we can count them as one constant (and in the meantime each such popped prefix and suffix is represented by at most four constants). 
Thus, $2n$ constants are popped in this way.
There are at most $4n$ crossing factors, see Lemma~\ref{lem: different crossing}, so in total $8n^2$ constants are introduced to the equation.

Consider constants initially present in the equation.
We show that for two such consecutive constants at least one is in a factor replaced during \algonephase.
Suppose otherwise and let $ab$ be those consecutive constants, clearly this factor is in $P$ computed by \algonephase.
If \algonephase{} compressed $ab$ during the compression of non-crossing factors,
then as we assumed that none of those $a$, $b$ was replaced,
this factor $ab$ was present and so it was compressed, contradiction.
So $ab$ was not compressed during the compression of non-crossing factors.
As it still was a factor when we began compression of crossing factors, it was considered as a factor to be replaced
and so we either replaced it or one of its constants was replaced,
which shows the claim.

\begin{figure}
	\centering
		\includegraphics{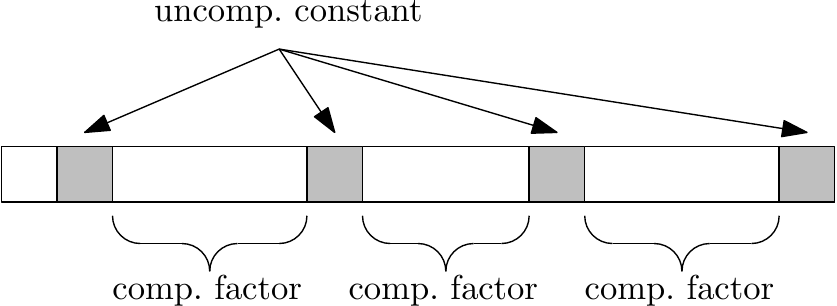}
	\caption{Each uncompressed letter is followed by a compressed factor}
	\label{fig:letters}
\end{figure}

Now, consider any sequence, say of length $k$, of constants initially present in the equation, see Fig.~\ref{fig:letters}.
We estimate, how many of its constants were removed during \algonephase.
Each constant (except perhaps the last one) that was not part of a replaced factor
can be associated with the replaced factor to its right.
As a factor is replaced with a single constant, this means that at least $\frac{k-1}{3}$ initially present constants were removed
(note that if a factor includes a constant popped from a variable then it looses all its initial constants,
as we count the constant that replaced it as the one popped from a variable).

Let $k_1$, $k_2$, \ldots, $k_\ell$ be the (maximal) sequences of constants initially present in the equation,
define $k = \sum_{i=1}^\ell k_i$ and observe that as there are two sides of the equation and each of (at most $n$)
variables terminates a sequence of constants, we have $\ell \leq n+2$.
Then at least
\begin{align*}
\sum_{i=1}^\ell \frac{k_i-1}{3}
	&=
\frac{k - \ell}{3}\\
	&\geq
\frac{k - n - 2}{3}
\end{align*}
constants initially present in the equation were removed.
On the other hand, as there were $8n^2$ new constants introduced, we conclude that the number of constants in the equation is at most
\begin{equation*}
k - \frac{k - n - 2}{3} + 8n^2 = \frac{2k}{3} + \frac{24n^2}{3} + \frac{n}{3} + \frac{2}{3} .
\end{equation*}
As initially the the equation had at most $27n^2$ constants, this yields that afterwards again it has at most $27n^2$ of them.

Lastly, note that every equation on the way has at most $35n^2$ constants: initially there are $27n^2$ of them and we pop at most $8n^2$ in total.
\qed
\end{proof}

\subsection{Proof of Lemma~\ref{lem: main} and generation of the graphs representation of all solutions}
We are now ready to give the proof of Lemma~\ref{lem: main}.
To this end we first reformulate it in the language of transformation of solutions.

\begin{lemma}[A modernised statement of Lemma~\ref{lem: main}]
\label{lem: main2}
Suppose that $(U_0,V_0)$ is a strictly proper equation with $|U_0|, |V_0| > 0$ and a simple solution $(\mysolution_0,\morphism_0)$.
Consider a run of \algonephase{} on $(U_0, V_0)$.
Then for some nondeterministic choices the obtained sequence of equations
$(U_0, V_0)$, $(U_1, V_1)$, \ldots, $(U_k, V_k)$ with corresponding families of inverse operators $\Phi_1$, $\Phi_2$, \ldots, $\Phi_k$
have simple solutions $(\mysolution_0,\morphism_0)$, $(\mysolution_1,\morphism_1)$, \ldots, $(\mysolution_k,\morphism_k)$
such that
\begin{itemize}
	\item $0 < k = \Ocomp(n^2)$;
	\item all $(U_0, V_0)$, $(U_1, V_1)$, \ldots, $(U_k, V_k)$ are proper and $(U_k,V_k)$ is strictly proper;
	\item $(U_i, V_i)$ with $(\mysolution_i, \morphism_i)$ is transformed to $(U_{i+1}, V_{i+1})$ with $(\mysolution_{i+1}', \morphism_{i+1})$,
	$\Phi_{i+1}$ is the corresponding family of inverse operators and $(\mysolution_{i+1},\morphism_{i+1})$ is a simplification of
	$(\mysolution_{i+1}', \morphism_{i+1})$.
\end{itemize}
\end{lemma}
\begin{proof}
Concerning the nondeterministic choices: firstly, let \algonephase{} correctly guess the set of crossing factors.
Then for each considered factor $ab$ from $P$ (say we have equation $(U_i, V_i)$ with the solution $(\mysolution_i, \morphism_i)$)
let it correctly guess, whether it is crossing or not in $\mysolution_i$.
Then by Lemma~\ref{lem: compression noncrossing},
$\algncr(U,V,ab)$ transforms  $(U_i, V_i)$ with $(\mysolution_i, \morphism_i)$ to
$(U_{i+1}, V_{i+1})$ with $(\mysolution_{i+1}', \morphism_{i+1})$, set $\mysolution_{i+1}$ as the simplification of $\mysolution_{i+1}'$.
The same lemma guarantees that $\weight(\mysolution_i,\morphism_i) > \weight(\mysolution_{i+1}',\morphism_{i+1})$,
if we replaced at least one factor in $(U_i, V_i)$.

Now, let \algonephase{} correctly guess that no pair in $P$ is non-crossing
(for the current equation $(U_\ell, V_\ell)$ with $(\mysolution_\ell, \morphism_\ell)$).
Let it also correctly guess the set of crossing factors.
Lastly, let during each call for $\algcr(U,V,ab)$ let it make the correct non-deterministic choices,
in the sense of Lemma~\ref{lem: compression crossing}.
Then this Lemma guarantees that the equation $(U_i, V_i)$ with $(\mysolution_i, \morphism_i)$ is transformed by $\algcr(U,V,ab)$
into $(U_{i+1}, V_{i+1})$ with $(\mysolution_{i+1}', \morphism_{i+1})$, set $\mysolution_{i+1}$ as the simplification of $\mysolution_{i+1}$.
The same lemma guarantees that $\weight(\mysolution_i,\morphism_i) \geq \weight(\mysolution_{i+1}',\morphism_{i+1})$
and the inequality is strict if $(U_i, V_i) \neq (U_{i+1}, V_{i+1}$.

By Lemma~\ref{lem: size bound} we know that the equation computed during $\algonephase(U, V)$
are proper and the last one of them is strictly proper, so this shows the bound on the size of $(U_i, V_i)$ and on $(U_k, V_k)$.

Concerning size of $k$: one equation is created for one compressed pair. There are $\Ocomp(n^2)$ such factors in $P$
and $\Ocomp(n)$ in $P'$.
To show that $k > 1$, it is enough to show that at least one pair is compressed.
Assume otherwise.
As at least one of $U$, $V$ is of length $2$ or more, there is some $ab \in P$.
If $ab$ is compressed as an element of $P$, we are done, otherwise it goes to $P'$.
Without loss of generality, let $ab$ be the first considered pair from $P'$.
Since it is in $P'$, it is crossing and so we fist uncross and then comprss it.
\qed
\end{proof}

\section{Running time for satisfiability}
For word equations over the free monoid (without the regular constraints)
the known algorithms~\cite{pr98icalp,wordequations} (non-deterministically)
verify the satisfiability in time polynomial in $n$ and $\log N$, where $N$ is the length of the length-minimal solution.
In particular, it is the common belief that $N$ is at most exponential in $n$,
and should this be so, those algorithms would yield that WordEquation is in \NP.
While our algorithm works in polynomial space, so far a similar bound on its running time is not known.

When no constraints are allowed in the equations the proof for the free monoid follows the lines similar to Lemma~\ref{lem: size bound}:
for a length-minimal solution \mysolution{} when $s$ is a factor in \sol U
then either $s$ is a factor of $U$ or it has a crossing occurrence in \mysolution{}
(as otherwise we could remove all factors $s$ from the solution, obtaining a shorter solution, which contradicts the length-minimality).
Thus $\algonephase(U, V)$ tries to compress each two-letter factor in \sol U
and so the same argument as in Lemma~\ref{lem: size bound} yields that the length of the length-minimal solution decreases after
$\algonephase(U, V)$ by a constant factor, so there are only $\log N$ applications of $\algonephase(U, V)$.

However, the regular constraints make such an argument harder:
when we cross out a factor $s$ from \sol X, the $\rho(\sol X)$ changes, which is not allowed.
However, this can be walked around: instead of crossing $s$ out we replace it with a single constant that has the same transition as $s$.
To this end we \emph{extend} the original alphabet: we add to the original alphabet $A$ constants $a_P$ for each $P \in \rho(A^+)$,
where $\rho(A^+)$ denotes the image of $A^+$ by $\rho$, \ie $\set{\rho(w)}{w \in A^+}$.
This set can be big, so we do not store it explicitly, instead we have a subprocedure that tests whether $P \in \rho(A^+)$.

There is another technical note: as we often apply simplification
we do not really know what happens with a length-minimal solution.
However, for a tuned definition of length-minimal solutions,
which takes into the account also the weight of the solution as a secondary factor, 
each length-minimal solution in some sense cannot be simplified.

\subsection{$\rho$-closed alphabets}
We begin with the precise definition of the $\rho$-closure of the alphabet and 
then show that a word equation with constraints is satisfiable over $A$ if and only if it
is satisfiable over the $\rho$-closure of $A$.

Given a finite alphabet $A$ together with a homomorphism $\rho$ from $A$ to $\Mn$
we say that an alphabet $A$ is $\rho$-\emph{closed} if $\rho(A) = \rho(A^+)$,
\ie for each word $w \in A^+$ there exists a constant $a$ such that $\rho(w) = \rho(a)$.

Usually, an alphabet is not $\rho$-closed, however, we can naturally extend with `missing' constants:
for an alphabet $A$ define a $\rho$-\emph{closure} $\cl_\rho(A)$ of $A$:
\begin{equation*}
\cl_\rho(A) = A \cup \set{a_P}{\text{there is } w \in A^+ \text{ such that } \rho(w) = P} \enspace ,
\end{equation*}
where each $a_P$ is a fresh constant not in $A$, $\inv{a_P} = a_{P^T}$ and $a_P \neq a_{P'}$ when $P \neq P'$.
It is easy to see that $\cl_\rho(A)$ is $\rho$-closed.
Whenever clear from the context, we will drop $\rho$ in the notation and talk about closure and $\cl$.

Viewing the equation over $A$ as an equation over $\cl(A)$ does not change the satisfiability.

\begin{lemma}
\label{lem: A iff cl(A)}
Suppose that we are given a word equation $(U, V)$ with regular constraints (defined using a homomorphism $\rho$)
over a free monoid generated by $\rho$-closed alphabet $A$.
Then $(U, V)$ has a solution over $A$ if and only if it has a solution when treated as an equation over the alphabet of constants $\cl(A)$.
\end{lemma}
Note that the set of all solution of the equation is of course different for $A$ and $\cl(A)$,
though in this section we are interested only in the satisfiability.
\begin{proof}
If $\mysolution$ is a solution over $A$ then it is of course a solution over $\cl(A)$.

On the other hand, when $\mysolution$ is a solution over $\cl(A)$ then we can create a solution over $A$:
for each $P \in \rho(A^+)$ choose a word $w_P$ such that $\rho(w_P) = P$,
moreover choose in a way so that $\inv{w_P} = w_{P^T}$, and replace every $a_P$ in \sol X by $w_P$.
Since constants $a_P$ do not occur in the equation, it is routine to check that the obtained substitution is a solution
(and since $\rho(w_P) = \rho(a_P)$, that all constraints are satisfied).
\qed
\end{proof}

\subsubsection*{Oracles for $\cl(A)$}
Note that the size of $\cl(A)$ may be much larger than $|A|$ (in fact, exponential in the input size).
Thus we cannot store it explicitly, instead, whenever a constant from $\cl(A) \setminus A$ is introduced to the instance,
we verify, whether it is indeed in $\cl(A)$, \ie whether the corresponding transition matrix $P$ is in $\rho(A^+)$.
In general, such check can be performed in \PSPACE{} (and in fact it is \PSPACE-complete in some cases),
but it can be performed more efficiently, when we know an upper-bound on $|\cl(A)|$.

\begin{lemma}
\label{lem: is P in closure}
It can be verified in \PSPACE, whether $P \in \rho(A)$.
Alternatively, this can  be verified in $\poly(|\rho(A^+)|, n)$ time.
\end{lemma}
\begin{proof}
The proof is standard.

Let $w_P = a_1 a_2 \cdots a_k$ be the shortest (non-empty) word such that $\rho(w_P) = P$.
As $\Mn$ has at most $2^{4m^2}$ elements, we have that $k \leq 2^{4m^2}$:
if it were longer then $\rho(a_1\cdots a_i) = \rho(a_1\cdots a_j)$ for some $i < j$ and thus
$\rho(a_1 a_2 \cdots a_k) = \rho(a_1 a_2 \cdots a_i a_{j+1} \cdots a_k)$, which cannot happen, as this word is shorter than $w_P$.

Thus in \PSPACE{} we can non-deterministically guess the constants $a_1$, $a_2$, \ldots, $a_k$ and verify that indeed $\rho(a_1 a_2 \cdots a_k) = P$.
Alternatively, we can deterministically list all elements of $\rho(A^+)$ in $\poly(|\rho(A^+)|, n)$ time.
\qed
\end{proof}

\subsubsection{Length-minimal solution}
We now give a proper definition of a length minimal solution:
First, we compare the solutions $(\mysolution_1,\morphism_1)$ and $(\mysolution_2,\morphism_2)$
by $|\mysolution_1(U)|$ and $|\mysolution_2(U)|$ and if those are equal,
by $\weight(\mysolution_1,\morphism_1)$ and $\weight(\mysolution_2,\morphism_2)$.

\begin{definition}[Length-minimal solution]
\label{def: length-minimal}
A solution $(\mysolution_1,\morphism_1)$ (of an equation $(U,V)$)
is \emph{length-minimal} if for every other solution $(\mysolution_2,\morphism_2)$ of this equation either
\begin{itemize}
	\item $|\mysolution_1(U)| < |\mysolution_2(U)|$ \emph{or}
	\item $|\mysolution_1(U)| = |\mysolution_2(U)|$ and $\weight(\mysolution_1,\morphism_1) \leq \weight(\mysolution_2,\morphism_2)$.	
\end{itemize}
\end{definition}

Note that our definition refines the usual one, in the sense that
if a solution is length-minimal according to Definition~\ref{def: length-minimal},
it is also length-minimal in the traditional sense, but not the other way around.
Furthermore, for the input equation the $\morphism_1$ is constant on all constants in the solution,
so our refined notion coincides with the traditional one.

\subsection{Equations over $\rho$-closed alphabet}
We can show that the length of the length-minimal solution shortens by a constant fraction in each run of \algonephase.

\begin{lemma}
\label{lem: length-minimal drops}
Let the original alphabet of the problem be a $\rho$-closed $A$.
Suppose that a strictly proper equation $(U,V)$ over an alphabet of constants $\letters \supseteq A$
has a length-minimal solution $(\mysolution, \morphism)$.
Then $(\mysolution, \morphism)$ is simple and for some non-deterministic choices \algonephase{} transforms $(U,V)$
with a $(\mysolution, \morphism)$ into a strictly proper $(U',V')$ with a simple solution $(\mysolution', \morphism')$
such that $|\mysolution'(X)| \leq \frac{2\sol X+1}{3}$ for each variable $X$.
\end{lemma}

Note that by definition a proper equation over $\letters$ has a homomorphism $\morphism : B \mapsto A^+$
that is compatible with $\rho$, \ie $\rho(b) = \rho(\morphism(b))$,
in particular $\rho(B) \subseteq \rho(A^+)$.

\begin{proof}
Consider a length-minimal solution $(\mysolution, \morphism)$ and an application of \algonephase{} on $(U, V)$.
According to Lemma~\ref{lem: main2} for appropriate non-deterministic choices made by \algonephase{}
we obtain a sequence of equations $(U, V) = (U_0, V_0)$, $(U_1, V_1)$, \ldots, $(U_k, V_k)$,
operators $\phi_1$, $\phi_2$, \ldots, $\phi_k$ and solutions $(\mysolution,\morphism) = (\mysolution_0,\morphism_0)$, \ldots, $(\mysolution_k, \morphism_k)$
such that
	$\mysolution_i = \phi_{i+1}[\mysolution_{i+1}']$
	and $(\mysolution_{i+1}, \morphism_{i+1})$ is a simplification of $(\mysolution_{i+1}, \morphism_{i+1})$.
Suppose first that $(\mysolution_0, \morphism_0)$ is not simple.
Then it uses a letter $b$ that is not present in $(U_0, V_0)$ and $\morphism_0(b) \in A^{\geqslant 2}$, let $P = \rho(b)$.
Consider a substitution $\mysolution_0'$ obtained from $\mysolution_0$ by replacing each $b$ by $a_P$
(and each $\inv b$ by $a_{P^T} = \inv {a_P}$).
As $B$ does not occur in $(U_0, V_0)$, $\mysolution_0'(U_0) = \mysolution_0'(V_0)$,
moreover $\rho(\mysolution_0'(X)) = \rho(\mysolution_0(X))$.
Since the alphabet of $\mysolution_0'$ is a subset of the alphabet of $\mysolution_0$,
we conclude that $(\mysolution_0',\morphism_0)$ is a solution of $(U_0, V_0)$.
Clearly $|\mysolution_0'(U_0)| = |\mysolution_0(U_0)|$.
Additionally, $|\mysolution_0'(X)| \leq |\mysolution_0(X)| $ and the inequality is strict if $\mysolution_0(X)$ contains $b$ or $\inv b$.
As for some $X$ the $\mysolution_0(X)$ indeed contains $b$ or $\inv b$, we conclude that $\weight(\mysolution_0',\morphism_0) < \weight(\mysolution_0',\morphism_0)$,
which contradicts the length-minimality of $(\mysolution_0, \morphism_0)$.

In a similar fashion we want to show that all $(\mysolution_1',\morphism_1)$, \ldots, $(\mysolution_k', \morphism_k)$ are simple.
For the sake of contradiction assume that this is not the case and 
take the smallest $i$ for which $(\mysolution_i', \morphism_i)$ is not simple.
Then in particular $(\mysolution_i, \morphism_i) \neq (\mysolution_i', \morphism_i)$
and $(\mysolution_i', \morphism_i)$ uses a constant $b$ that does not occur in the alphabet of $(U_i, V_i)$, let $P = \rho(b)$.
By discussion between the lemma and the proof, there is $a_P \in A$ such that $\rho(a_P) = P$.
Create $(\mysolution_i'', \morphism_i)$ by replacing each $b$ and $\inv b$ in any \sol X by $a_P$ and $\inv{a_P}$.
As in the case of $(\mysolution_0',\morphism_0)$ it is easy to verify that
$(\mysolution_i'',\morphism_i)$ is a solution of $(U_i, V_i)$.
Now, by definition of $\mysolution_i$ and $\phi_1, \ldots, \phi_i$
$$
\mysolution_0 = \phi_{1} \circ \phi_2 \circ \cdots \circ \phi_i[\mysolution_i'].
$$
and denote by $\phi$ the $\phi_{1} \circ \phi_2 \circ \cdots \circ \phi_i$.
Consider an action of any of $\phi_{1}$, $\phi_2$, \ldots, $\phi_i$ on some substitution.
It may append and prepend letters to \sol X (independently of $X$ and of the substitution)
and it may replace some letters (outside of $A$) by longer factors, again independently of $X$ and of the substitution.
Thus also their composition $\phi$ has this property.
Consider now 
$$
\mysolution_0'' = \phi[\mysolution_i''] .
$$
We intend to show that $(\mysolution_0'',\morphism_0)$ is a solution of $(U_0, V_0)$
and that it contradicts the length-minimality of $(\mysolution_0, \morphism_0)$.
We need to show that $\morphism_0$ is defined on any letter assigned by $\mysolution_0''$ outside $A$.
But if this was the case, the same letter would be used also by $\mysolution_0$:
if this letter was used by $\mysolution_i''$ and not replaced by $\phi$ then the same applies to $\mysolution_i'$.
If it was appended or prepended, then the same letter is appended or prepended to $\mysolution_i'$.
It is left to show that $(\mysolution_0, \morphism_0)$ is not length-minimal.

Firstly, we show that $|\mysolution_0(U_0)| \geq |\mysolution_0''(U_0)|$.
Consider that $\phi[\mysolution_i']$ and $\phi[\mysolution_i'']$ and their action of $X$.
Then $\phi$ prepends and appends the same strings to $\mysolution_i(X)$ and $\mysolution_i''(X)$,
additionally, it replaces letters (outside $A$) in $\mysolution_i(X)$ and $\mysolution_i''(X)$ by the same strings.
As $\mysolution_i''(X)$ is obtained from $\mysolution_i(X)$ by replacing $b$ and $\inv b$ by $a_P, \inv{a_P} \in A$,
so $\mysolution_0''(X)$ and $\mysolution_0(X)$ differ only in strings that replace $b$ and $\inv b$:
in the former those are $a_P$ and $\inv{a_P}$ while in the latter those are some strings (of lengths at least $2$).
Thus $|\mysolution_0''(X)| \leq |\mysolution_0(X)|$ and the inequality is strict when $b$ or $\inv b$ is in $\mysolution_i'(X)$.
As this constant occurs in at least one $\mysolution_i'(X)$, we conclude that $\mysolution_0$ is not length-minimal.

In exactly the same way we can show that if for some $i$ the $ab$ is a factor in $\mysolution_i(U_i)$ then some $ab$-factor occurs in $(U_i, V_i)$
or is crossing for $\mysolution_i$: otherwise we could replace each $ab$-factor $s$ in any \sol X with $a_{\rho(s)}$,
the obtained substitution $\mysolution_i''$ (together with $\morphism_i$)
is a solution and $\phi_{1} \circ \phi_2 \circ \cdots \circ \phi_i[\mysolution_i'']$ (together with $\morphism_0$) is a solution of $(U_0, V_0)$
and this solution contradicts the length-minimality of $\mysolution_0 = \phi_{1} \circ \phi_2 \circ \cdots \circ \phi_i[\mysolution_i]$.

We move to the main part of the proof.
Firstly, assume that \algonephase{} correctly guesses the set of crossing factors at the very beginning.
Let $(U_\ell, V_\ell)$ be the equation obtained when \algonephase{} (correctly) decides that no pair from $P$ is non-crossing
and afterwards it correctly lists crossing factors (and begins to compress them).
Consider $\mysolution_k(X)$. As $\mysolution_0(X) = \phi_{1} \circ \phi_2 \circ \cdots \circ \phi_k[\mysolution_k]$
and each of the operator can append constants, prepend constants and replace constants by (strictly) longer words,
we know that each constant in $\mysolution_k(X)$  corresponds either to a single (``uncompressed'') constant from
$\mysolution_0(U_0)$ or to a longer word inside $\mysolution_0(U_0)$ (compressed into this constant).
We claim that in $\mysolution_k(X)$ there are no two consecutive constants that are uncompressed.
Using this we can easily show the claim: each uncompressed constant in $\mysolution_k(U_k)$, except perhaps the last,
is followed by a constant representing at least two constants in the initial word.

Suppose that $ab$ is a factor in $\mysolution_k(U_k)$ and that they both are uncompressed.
Thus the corresponding $ab$ is present in $\mysolution_\ell(U_\ell)$
and so be earlier claim $ab$ either occurs in $(U_\ell, V_\ell)$ or is a crossing factor,
in either case it will be in $P'$.
Then \algonephase{} performs the $ab$ compression, contradiction.
\qed
\end{proof}

\subsubsection{Running time for equations over groups}
As a consequence, we can verify the  satisfiability of a word equation in free groups (without rational constraints)
in (nondeterministic) time $\npoly(\log N, n)$, where $N$ is the size of the length minimal solution.

\begin{theorem}
The satisfiability of word equation over free group (without rational constraints) can be verified in \PSPACE{}
and at the same time $\npoly(\log N , n )$ time, where $N$ is the size of the length-minimal solution 
of this equation.
\end{theorem}
\begin{proof}
We reduce the problem in a free group to the corresponding one in a free semigroup, see \prref{prop:dghftom05}.
In this way we introduce regular constraint, and these are the only constraints in the problem.
This constraint says that $a\inv a$ cannot be a factor of $X$, for any $a$.
The NFA for this condition has $|\Gamma|^2 + 2$ states:
\begin{itemize}
	\item sink (all transitions to itself)
	\item initial state 
	\item a state $(a,b)$, where $a$ is the first constant of the word and $b$ the last.
\end{itemize}
The transitions are obvious.
It is easy to see that $\Mn$ have $\Ocomp(|\Gamma|)$ elements.
Thus the subprocedure for checking whether $P \in \rho(A^+)$ can be implemented in $\poly(n)$,
see Lemma~\ref{lem: is P in closure}.

Concerning the problem in the free monoid, we first extend the alphabet $A$ to $\cl(A)$.
By Lemma~\ref{lem: A iff cl(A)} those problems are equisatisfiable.
By Lemma~\ref{lem: length-minimal drops} the length of the substitution for a variable drops by a constant fraction after each application
of \algonephase{} (for appropriate non-deterministic choices),
so there are only $\Ocomp(\log N)$ application of this procedure till all variables are removed.
As there are no variables in the equation and the weight (\ie length of the equation) decreases after each application of \algonephase,
afterwards there are only $\Ocomp(n^2)$ such applications.
Clearly each such an application takes time polynomial in $n$ (as all equation are proper by Lemma~\ref{lem: size bound}).
\qed
\end{proof}

\section{Applications}
Using the results above we obtain the following theorem: 
\begin{theorem}
It can be decided in \PSPACE{} whether the input system with rational constraints has a finite number of solutions.
\end{theorem}
\begin{proof} To find out whether the equation has infinite number of 
solutions it is enough to find a path from a start node of the graph
to a final node which either
\begin{itemize}
	\item contains a loop \emph{or}
	\item one of the edges of the path is labeled by a linear system of equations having infinite number of solutions \emph{or}
	\item the final node has infinite number of solutions, which means that it is of the form $(X, Y)$, $\rho(X) = \rho(Y)$
	and there are infinitely many words $w$ such that $\rho(w) = \rho(X)$.
\end{itemize}
The first condition is a simple reachability in a graph, which can be performed in \NPSPACE,
as the description of the nodes and edges are of polynomial size.
The second condition can be verified in \NP, see Proposition~\ref{prop:lDs}.
The last condition can be easily verified in \PSPACE. Since \NPSPACE{} contains \NP{} and is equal to \PSPACE,
the search of such a path can be done in \PSPACE.

Now if none of those conditions is satisfied, the graph representation of all solutions is a finite DAG,
for each edge the family of inverse operators is finite and each final node has finitely many solutions,
which implies that there are only finitely many solutions in total.
\qed
\end{proof}

\newcommand{\Ch}{Ch}
\newcommand{\Yu}{Yu}


\begin{thebibliography}{10}

\bibitem{ben69}
M.~Benois.
\newblock Parties rationelles du groupe libre.
\newblock {\em C. R. Acad. Sci. Paris, S{\'e}r. A}, 269:1188--1190, 1969.

\bibitem{DahmaniGui10}
F.~Dahmani and V.~Guirardel.
\newblock Foliations for solving equations in groups: free, virtually free and
  hyperbolic groups.
\newblock {\em J. of Topology}, 3:343--404, 2010.

\bibitem{dickson1913}
L.~E. Dickson.
\newblock Finiteness of the odd perfect and primitive abundant numbers with
  {$n$} distinct prime factors.
\newblock {\em American Journal of Mathematics}, 35(4):413--422, 1913.

\bibitem{dgh05IC}
V.~Diekert, C.~Guti{\'e}rrez, and {\Ch}.~Hagenah.
\newblock The existential theory of equations with rational constraints in free
  groups is {PSPACE}-complete.
\newblock {\em Information and Computation}, 202:105--140, 2005.
\newblock Conference version in STACS 2001, LNCS 2010, 170--182, 2004.

\bibitem{DiekertLohrey08}
V.~Diekert and M.~Lohrey.
\newblock Word equations over graph products.
\newblock {\em IJAC}, 18(3):493--533, 2008.

\bibitem{dmm99tcs}
V.~Diekert, {\Yu}.~Matiyasevich, and A.~Muscholl.
\newblock Solving word equations modulo partial commutations.
\newblock {\em Theoretical Computer Science}, 224:215--235, 1999.
\newblock Special issue of LFCS'97.

\bibitem{dm06}
V.~Diekert and A.~Muscholl.
\newblock Solvability of equations in free partially commutative groups is
  decidable.
\newblock {\em International Journal of Algebra and Computation},
  16:1047--1070, 2006.
\newblock Journal version of ICALP 2001, 543--554, LNCS 2076.

\bibitem{dur95}
V.~G. Durnev.
\newblock Undecidability of the positive $\forall\exists^{3}$-theory of a free
  semi-group.
\newblock {\em Sibirsky Matematicheskie Jurnal}, 36(5):1067--1080, 1995.
\newblock In Russian; English translation: {\em Sib. Math. J., 36\/}(5),
  917--929, 1995.

\bibitem{eil74}
S.~Eilenberg.
\newblock {\em Automata, Languages, and Machines}, volume~A.
\newblock Academic Press, New York and London, 1974.

\bibitem{gut98focs}
C.~Guti{\'e}rrez.
\newblock Satisfiability of word equations with constants is in exponential
  space.
\newblock In {\em Proc. 39th Ann.~Symp. on Foundations of Computer Science
  (FOCS'98), Los Alamitos (California)}, pages 112--119. IEEE Computer Society
  Press, 1998.

\bibitem{gut2000stoc}
C.~Guti{\'e}rrez.
\newblock Satisfiability of equations in free groups is in {PSPACE}.
\newblock In {\em Proceedings 32nd Annual ACM Symposium on Theory of Computing,
  STOC'2000}, pages 21--27. ACM Press, 2000.

\bibitem{HU}
J.~E. {Hopcroft} and J.~D. {Ulman}.
\newblock {\em Introduction to Automata Theory, Languages and Computation}.
\newblock Addison-Wesley, 1979.

\bibitem{IlPl}
L.~Ilie and W.~Plandowski.
\newblock Two-variable word equations.
\newblock {\em Theoretical Informatics and Applications}, 34:467--501, 2000.

\bibitem{wordequations}
A.~Je\.z.
\newblock {Recompression: a simple and powerful technique for word equations}.
\newblock In N.~Portier and T.~Wilke, editors, {\em STACS}, volume~20 of {\em
  LIPIcs}, pages 233--244, Dagstuhl, Germany, 2013. Schloss
  Dagstuhl--Leibniz-Zentrum fuer Informatik.

\bibitem{KMII98}
O.~{Kharlampovich} and A.~{Myasnikov}.
\newblock Irreducible affine varieties over a free group. {II}: {S}ystems in
  triangular quasi-quadratic form and description of residually free groups.
\newblock {\em J. of Algebra}, 200(2):517--570, 1998.

\bibitem{KMIV06}
O.~{Kharlampovich} and A.~{Myasnikov}.
\newblock Elementary theory of free non-abelian groups.
\newblock {\em J. of Algebra}, 302:451--552, 2006.

\bibitem{kp96}
A.~Ko{\'s}cielski and L.~Pacholski.
\newblock Complexity of {M}akanin's algorithm.
\newblock {\em Journal of the Association for Computing Machinery},
  43(4):670--684, 1996.

\bibitem{mak77}
G.~S. Makanin.
\newblock The problem of solvability of equations in a free semigroup.
\newblock {\em Math. Sbornik}, 103:147--236, 1977.
\newblock English transl. in Math. USSR Sbornik 32 (1977).

\bibitem{mak82}
G.~S. Makanin.
\newblock Equations in a free group.
\newblock {\em Izv. Akad. Nauk SSR}, Ser. Math. 46:1199--1273, 1983.
\newblock English transl. in Math. USSR Izv. 21 (1983).

\bibitem{mak84}
G.~S. Makanin.
\newblock Decidability of the universal and positive theories of a free group.
\newblock {\em Izv. Akad. Nauk SSSR}, Ser. Mat. 48:735--749, 1984.
\newblock In Russian; English translation in: {\em Math.~USSR Izvestija, 25},
  75--88, 1985.

\bibitem{mat93}
{\Yu}.~Matiyasevich.
\newblock {\em Hilbert's Tenth Problem}.
\newblock MIT Press, Cambridge, Massachusetts, 1993.

\bibitem{mat97lfcs}
{\Yu}.~Matiyasevich.
\newblock Some decision problems for traces.
\newblock In S.~Adian and A.~Nerode, editors, {\em Proceedings of the {4th}
  International Symposium on Logical Foundations of Computer Science (LFCS'97),
  Yaroslavl, Russia, July 6--12, 1997}, volume 1234 of {\em Lecture Notes in
  Computer Science}, pages 248--257, Heidelberg, 1997. Springer-Verlag.
\newblock Invited lecture.

\bibitem{Pl2}
W.~Plandowski.
\newblock Satisfiability of word equations is in {NEXPTIME}.
\newblock In {\em Proceedings of the Symposium on the Theory of Computing
  STOC'99}, pages 721--725. ACM Press, 1999.

\bibitem{pla04jacm}
W.~Plandowski.
\newblock Satisfiability of word equations with constants is in {PSPACE}.
\newblock {\em Journal of the Association for Computing Machinery},
  51:483--496, 2004.

\bibitem{Pl5}
W.~Plandowski.
\newblock An efficient algorithm for solving word equations.
\newblock In {\em Proceedings of the 38th Annual Symposium on Theory of
  Computing STOC'06}, pages 467--476. ACM Press, 2006.

\bibitem{Pl6}
W.~Plandowski.
\newblock personal communication, 2014.

\bibitem{pr98icalp}
W.~Plandowski and W.~Rytter.
\newblock Application of {L}empel-{Z}iv encodings to the solution of word
  equations.
\newblock In K.~G. Larsen et~al., editors, {\em Proc. 25th International
  Colloquium Automata, Languages and Programming (ICALP'98), Aalborg (Denmark),
  1998}, volume 1443 of {\em Lecture Notes in Computer Science}, pages
  731--742, Heidelberg, 1998. Springer-Verlag.

\bibitem{raz87}
A.~A. Razborov.
\newblock {\em On Systems of Equations in Free Groups}.
\newblock PhD thesis, Steklov Institute of Mathematics, 1987.
\newblock In Russian.

\bibitem{raz93}
A.~A. Razborov.
\newblock On systems of equations in free groups.
\newblock In {\em Combinatorial and Geometric Group Theory}, pages 269--283.
  Cambridge University Press, 1994.

\bibitem{rs95}
E.~Rips and Z.~Sela.
\newblock Canonical representatives and equations in hyperbolic groups.
\newblock {\em Inventiones Mathematicae}, 120:489--512, 1995.

\bibitem{sch91}
K.~U. Schulz.
\newblock {M}akanin's algorithm for word equations --- {T}wo improvements and a
  generalization.
\newblock In K.~U. Schulz, editor, {\em Word Equations and Related Topics},
  volume 572 of {\em Lecture Notes in Computer Science}, pages 85--150,
  Heidelberg, 1991. Springer-Verlag.

\end{thebibliography}
\end{document}